\documentclass[usenix]{siistemplate}

\usepackage{siiscmds}                           
\usepackage[ruled,linesnumbered]{algorithm2e}    
\usepackage{algpseudocode}
\SetCommentSty{itshape}
\usepackage{booktabs}                            
\usepackage{centernot}                           
\usepackage{colortbl}                            
\usepackage{graphicx}                            
\usepackage{hhline}                              
\usepackage{lipsum}                              
\usepackage{mathtools}                           
\usepackage[all]{nowidow}                        
\usepackage{optidef}                             
\usepackage{pgfplotstable}                       
\usepackage{subcaption}                          
\usepackage[binary-units]{siunitx}                             
\usepackage{substr}                              
\usepackage{tikz}                                
\usepackage{tcolorbox}                           
\usepackage{xcolor}                              
\usepackage{circledsteps}
\definecolor{light-gray}{gray}{0.968}

\usepackage{xspace}                              
\usepackage{glossaries}
\usepackage{multirow}
\usepackage{float}
\usepackage{url}
\usepackage{pgf}
\usepackage{wrapfig}
\usepackage{pifont}
\usepackage{epsfig}
\usepackage{endnotes}
\usepackage{textcomp}
\usepackage{comment}
\usepackage{enumitem}
\usepackage{microtype}
\usepackage{tabularx}
\usepackage{amsthm}

\newtheorem{theorem}{Theorem}

\newcommand{\newtxt}[1]{#1}
\newcommand{\oldtxt}[1]{}

\papertitle{On Scalable Integrity Checking for\\Secure Cloud Disks}
\addauthors{%
    {Quinn Burke, Ryan Sheatsley, Rachel King, Owen Hines, Michael Swift, and Patrick McDaniel}/ /{University of Wisconsin-Madison}
    //University of Wisconsin-Madison,
}

\begin{document}
\ifthenelse{\equal{\template}{ccs}\OR{}\equal{\template}{pets}\OR{}\equal{\template}{sosp}}{%
    \begin{abstract}{Merkle hash trees are the standard method to protect the integrity and freshness of stored data. However, hash trees introduce additional compute and I/O costs on the I/O critical path, and prior efforts have not fully characterized these costs. In this paper, we quantify performance overheads of storage-level hash trees in realistic settings.  We then design an optimized tree structure called \textit{Dynamic Merkle Trees (DMTs)} based on an analysis of root causes of overheads. DMTs exploit patterns in workloads to deliver up to a $2.2\times$ throughput and latency improvement over the state of the art. Our novel approach provides a promising new direction to achieve integrity guarantees in storage efficiently and at scale.
}\end{abstract}
    \maketitle
}{%
    \maketitle
    \begin{abstract}Merkle hash trees are the standard method to protect the integrity and freshness of stored data. However, hash trees introduce additional compute and I/O costs on the I/O critical path, and prior efforts have not fully characterized these costs. In this paper, we quantify performance overheads of storage-level hash trees in realistic settings.  We then design an optimized tree structure called \textit{Dynamic Merkle Trees (DMTs)} based on an analysis of root causes of overheads. DMTs exploit patterns in workloads to deliver up to a $2.2\times$ throughput and latency improvement over the state of the art. Our novel approach provides a promising new direction to achieve integrity guarantees in storage efficiently and at scale.
\end{abstract}
}

\section{Introduction}
\label{introduction}

An increasing number of attacks against cloud services has fueled significant investment and research into trusted cloud storage systems: systems that provide high assurance of the confidentiality and integrity of data stored in-memory and on-disk through hardware-based access-controls and cryptographic proof systems~\cite{tsai_graphene-sgx_nodate,priebe_enclavedb_2018,arasu2021fastver}. To this end, using a Merkle hash tree~\cite{merkle1989certified} has become the state-of-the-art method to protect the integrity and \textit{freshness} of both volatile~\cite{taassori2018vault,feng2021scalable} and persistent~\cite{arasu2021fastver,chakraborti2017dm} storage. An exemplar use case is protecting disks attached to confidential virtual machines~\cite{aws_sev_snp,gcp_confidential_vm,azure_confidential_vm}.

However, hash trees introduce additional compute (hashing) and I/O (metadata fetching) costs on the I/O critical path, which can severely degrade performance. For example, consider that a 1~TB disk contains $\approx$268~M 4~KB blocks. A typical balanced, binary hash tree over the disk blocks would have a height of 28, requiring (at least) 28 hashes to be computed on every read or write. The total cost of fetching metadata and verifying/updating hashes can exceed several hundred $\mu s$, dwarfing the baseline latency of performing a data access on a high-performance storage device (which can be $<60~\mu s$).

Prior works have studied this phenomenon, primarily in the context of secure volatile memory~\cite{taassori2018vault,yan2006improving,gassend2003caches,feng2021scalable}. However, their performance implications in the context of (cloud) block storage at large, and low-latency storage devices in particular, remain largely unknown. The key difference is that storage devices are subject to vastly different workload characteristics, capacities, and cache behaviors than memory devices.

In this paper, we take a first-principles approach to analyzing hash tree performance in the context of cloud block storage. First, we demonstrate that state-of-the-art hash tree designs incur significant overheads and fail to reliably scale to large disk capacities. We then demonstrate that hashing (CPU) costs are the primary performance bottleneck. Next, we address the challenge of reducing hashing costs by posing the fundamental question: \textit{How can we model an optimal hash tree for cloud block storage?} 
We show that the problem of finding an optimal hash tree can be reduced to the problem of finding an optimal prefix tree in the context of lossless data compression~\cite{huffman1952method}. More specifically, by constructing a hash tree as an optimal prefix code, we can produce a hash tree that achieves optimal throughput under a known workload profile.

Building on our observations of optimal trees, we then develop an \textit{online} solution that can approximate an optimal tree (without a priori knowledge) by learning and adapting to workload patterns on-the-fly. It is known that real-world workloads are characterized by skewed access patterns (i.e., where a small number of blocks are accessed much more frequently than others) across all layers of the memory hierarchy~\cite{li2023depth,wang2023locality,cooper2010benchmarking,yang2016write,arasu2021fastver}. In an offline setting, this manifests as optimal hash trees often being far from balanced---where frequently accessed blocks have shorter verification/update paths in the tree than infrequently accessed blocks. Towards this, we introduce a novel dynamic, unbalanced hash tree design called \textit{Dynamic Merkle Trees (DMTs)}. DMTs are based on the splay trees commonly used in garbage collection and IP routing~\cite{sleator1985self}, and they self-adjust at runtime to reduce hashing costs for frequently accessed data.

We implemented DMTs and evaluated them in a real cloud setting with AWS EC2 instances and NVMe devices. We performed a broad performance analysis across a range of system and workload settings, parameterized by disk capacity, read/write ratio, I/O size, etc. Using a set of Zipfian workloads, an Alibaba dataset, and a Filebench OLTP workload, we show that the static nature of state-of-the-art approaches becomes prohibitive: they deliver less than 50\% of optimal throughput on average across all experiments. In contrast, DMTs capitalize on skewed access patterns, delivering >85\% of optimal throughput and up to a $2.2\times$ throughput and latency improvement over the state of the art.

We conclude that for cloud block storage, balanced trees are ill-suited as the base construction of a hash tree, and DMTs are a preferable alternative. DMTs provide a new foundation in the search for integrity mechanisms that \newtxt{operate efficiently at scale. Our code is plug-and-play into Linux and open-sourced at \href{https://github.com/MadSP-McDaniel/dmt}{https://github.com/MadSP-McDaniel/dmt}.} \oldtxt{that can operate efficiently at scale. Our code is plug-and-play into standard Linux systems and open-sourced at [Anonymized link].}

\section{Background}
\label{sec:bg}

\shortsection{Cloud Block Storage}
Block storage is a backbone of modern public cloud infrastructure~\cite{aws-ebs,gc-pdisk,azure-mdisk}. While there are various deployment models for cloud applications and storage, we consider a standard Infrastructure-as-a-Service (IaaS) deployment where an application runs inside of a guest VM and reads and writes to a fast, local NVMe disk attached to the VM (\autoref{fig:system-model}). The application may be end-user facing (e.g., a web server) or the last hop in a networked storage system (e.g., a file server for other cloud-hosted applications).

\begin{figure}[t]
    \centering
    \includegraphics[width=0.9\linewidth]{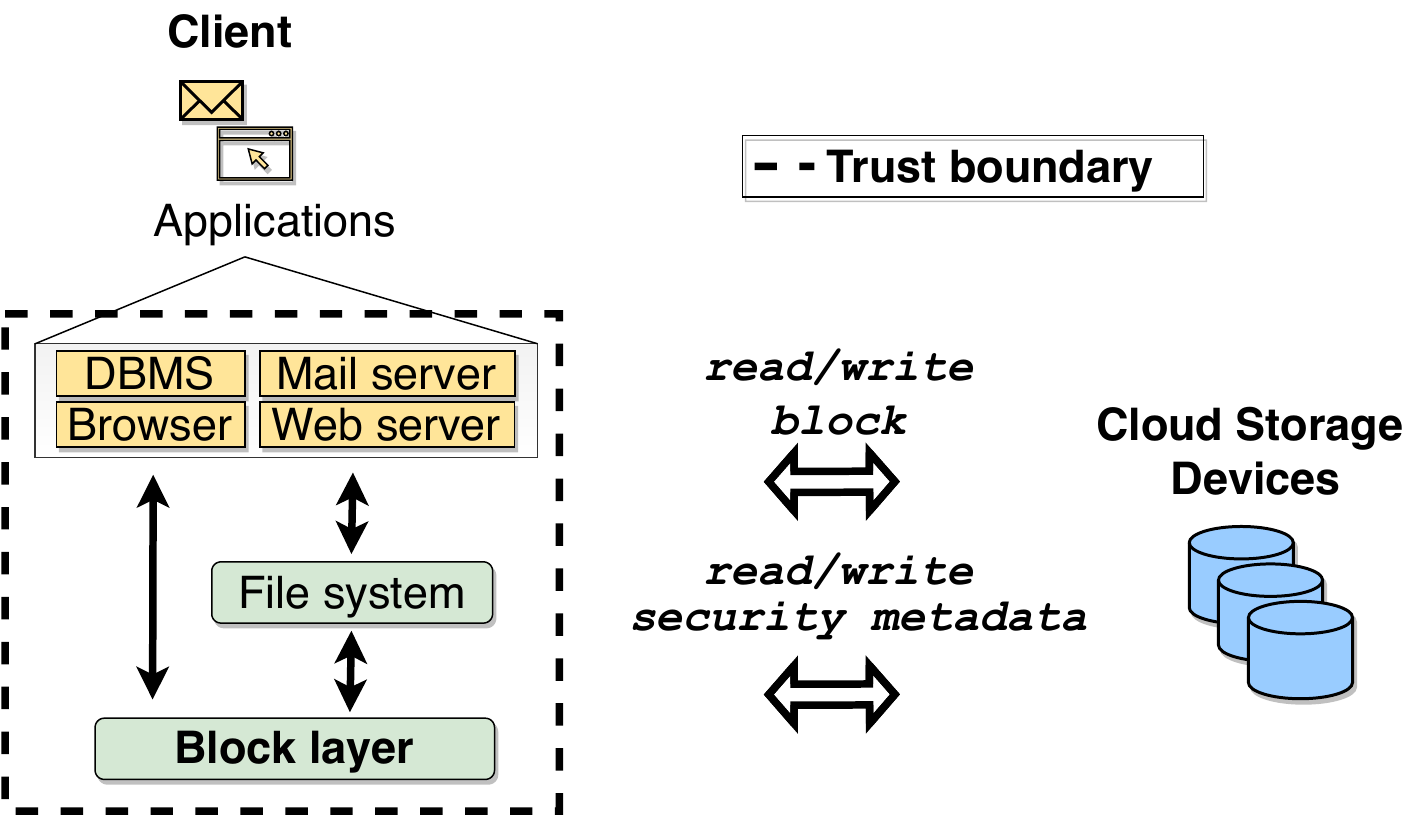}
    \caption{We assume that VM memory contents are trusted and cloud storage devices are untrusted; VM memory can be protected with trusted execution primitives~\cite{aws_sev_snp}.}
    \label{fig:system-model}
\end{figure}

\shortsection{Merkle Hash Trees} 
Merkle hash trees are the state-of-the-art method to protect the integrity and freshness of arbitrary datasets---largely due to their proven theoretical efficiency~\cite{merkle1989certified,gassend2003caches,mckeen_innovative_2013,tsai_graphene-sgx_nodate,priebe_sgx-lkl_2020}. They have played a pivotal role in ensuring boot disk integrity with Linux \textit{dm-verity}~\cite{android-dm-verity}---where they are implemented as a custom device driver that intercepts I/Os and implements the hash tree logic.

As shown in~\autoref{fig:merkle-hash-tree}, a Merkle hash tree (or more simply, a hash tree) is typically a balanced binary tree, with each node in the tree containing a hash value. A leaf node contains the hash (MAC) of a data block (and a cipher IV when encrypting data), and an internal node contains the hash of the concatenation of the hashes of its two children. Internal node hashes are iteratively computed from leaf to root. The root hash \textit{authenticates} the current contents of the storage device and is typically stored in a secure location (e.g., a persistent on-chip register or a TPM~\cite{perez2006vtpm,taassori2018vault}). All other nodes in the tree are stored on disk alongside the data. The number of leaf nodes in the tree $n$ is equal to the number of blocks on the storage device, and the total number of tree nodes is $2n-1$. 

\begin{figure}[t]
    \centering
    \includegraphics[width=0.9\linewidth]{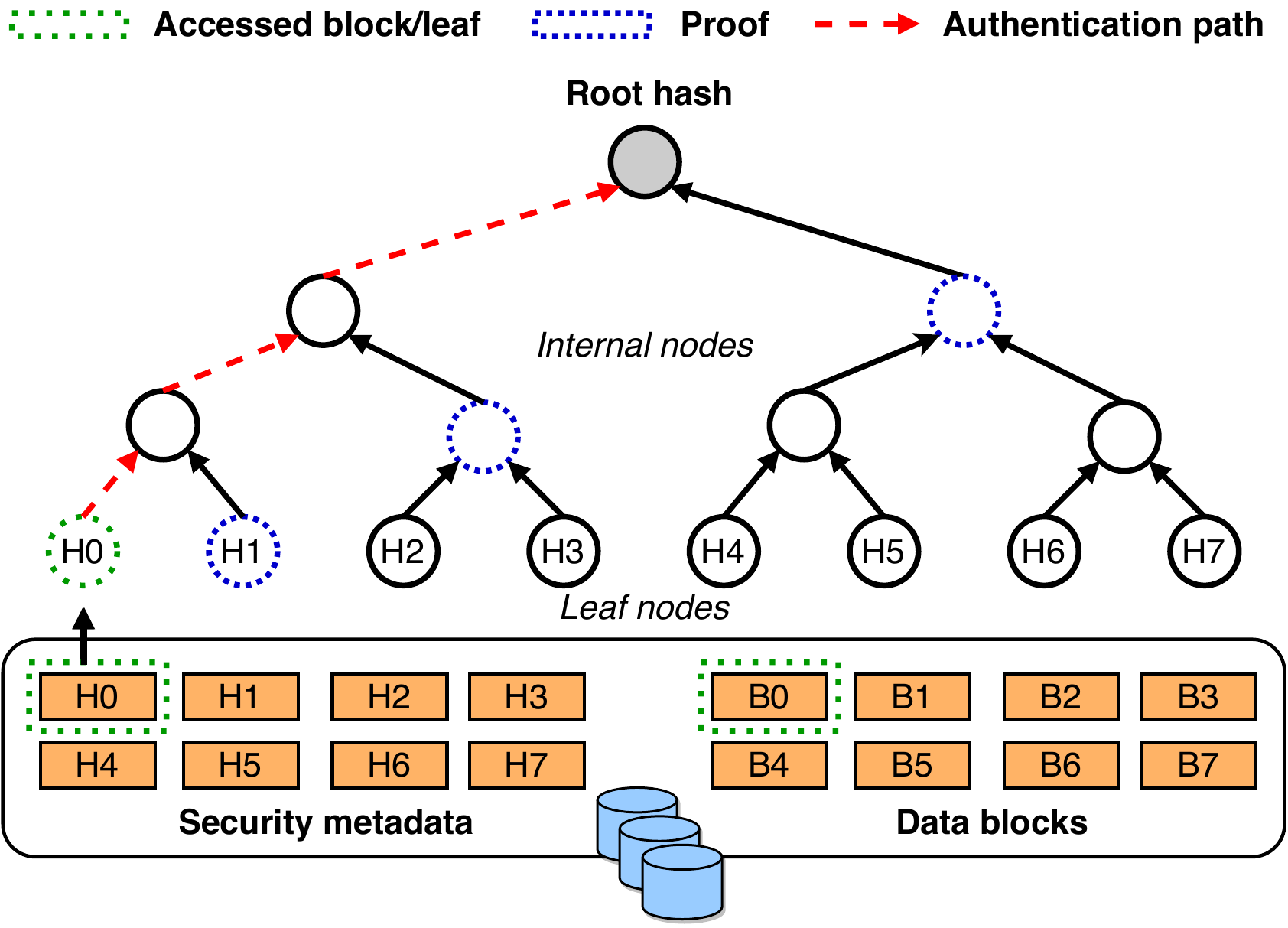}
    \caption{A Merkle hash tree protects the integrity and freshness of data read from/written to a storage device.}
    \label{fig:merkle-hash-tree}
\end{figure}

There are two primitive operations on a hash tree: \textit{verification} and \textit{update}. When a block is read, it must be verified against the root hash. The client's block layer first fetches the (encrypted) block data, MAC, and cipher IV from disk. It checks that the retrieved MAC is consistent with the retrieved block data by rehashing the data and comparing. It then fetches the proof of authenticity, a set of sibling hashes along the path from the accessed leaf to the root (see nodes highlighted in blue). The retrieved MAC is inserted into the tree at the appropriate leaf position, and parent hashes are iteratively computed along the \textit{authentication path} using the sibling hashes (see red arrows). The computed root hash is then compared against the known root hash. If the two hashes match, verification succeeds. When a block is written, a new hash must be computed (\texttt{H0}) and the hash tree updated (\texttt{H0} ancestors). Updates are handled similar to verification, with a new root hash computed and saved to the secure location.

Caching hashes in secure memory (i.e., in a protected memory region) is also a standard hash tree optimization~\cite{gassend2003caches,arasu2021fastver}. Caching reduces I/O costs associated with fetching hashes during a verification or update. It also enables early returns when fetching and verifying hashes; cached hashes were already authenticated, so any hash that has been tampered with will result in a failed check at some level in the tree.

\section{Security Model}
\label{sec:secmodel}

\textit{Data-only} attacks---attacks based on maliciously crafted data rather than control flow hijacking---have been recently shown to present a significant threat to modern applications~\cite{johannesmeyer2024practical,kurmus2017random,galanou2023trustworthy,bohling2020subverting,hu2016data}. Any attack that can be launched from (untrusted) storage would fundamentally upend the security guarantees provided by the rest of the system. Below we describe these attacks and outline security requirements to mitigate them.

\shortsection{Trust Model} 
We assume all VM contents (code and data) are trusted and the storage devices are untrusted. VM memory contents can be protected with hardware-based isolation primitives such as AMD SEV-SNP~\cite{aws_sev_snp}. The trusted and untrusted components therefore have a simple block read/write interface (\autoref{fig:system-model}). This models untrusted disks attached to, for example, confidential virtual machines~\cite{aws_sev_snp,gcp_confidential_vm,azure_confidential_vm}. 

\shortsection{Threat Model} 
We consider a privileged attacker who has access to the hypervisor or storage backbone in a public cloud datacenter~\cite{tsai_graphene-sgx_nodate,arnautov_scone_nodate}. This could be a malicious co-tenant who was able to escalate privilege, or a malicious cloud administrator. The attacker has the ability to access, corrupt, swap, drop, record, inject, or replay any data across the storage backbone. 

\textit{Example attacks.}
Armed with the capability to inject arbitrary data into the storage interface, the attacker could replay old data to the VM~\cite{kurmus2017random,johannesmeyer2024practical}. Data would bubble up the call stack and either cause the VM to deliver old data to applications, or cause an outdated version of a binary to be read from disk and executed~\cite{johannesmeyer2024practical,galanou2023trustworthy,bohling2020subverting}. Similarly, consider an ext4 file system formatted on top of the disk. An attacker could arbitrarily replay inode table blocks and cause the VM OS to recognize an invalid set of permissions on a file, enabling unauthorized access to the file. Checksums or keyed hashes alone cannot prevent these data-only attacks: the received data would still pass verification.

\shortsection{Security Requirements} 
Ensuring the safety of user data and correct execution of applications therefore requires three security properties for storage: authenticity, uniqueness, and freshness~\cite{avanzi2022cryptographic}. Keyed encryption and MACs can ensure confidentiality, authenticity (prevents corruptions), and uniqueness (prevents relocation attacks). Merkle hash trees then ensure data freshness~\cite{gassend2003caches,arasu2021fastver,angel2023nimble,matetic2017rote}: the root hash reflects the current \textit{version} of the storage device, so a replay attack would require changing the root hash, which is stored in a secure location and is out of control of the attacker~\footnote{To improve performance, some prior works have loosened security requirements by permitting lazy verification~\cite{arasu2021fastver}. However, this violates freshness guarantees; we therefore do not consider lazy verification in our analysis.}.

\section{Motivation}
\label{sec:motivation}
Though hash trees have played a pivotal role in ensuring boot disk integrity with Linux \textit{dm-verity}~\cite{android-dm-verity}, their performance implications in the context of (cloud) block storage at large, and low-latency storage devices in particular, are largely unknown. In fact, prior works have identified that hash tree overheads can severely degrade performance for secure memory systems, and optimizations abound~\cite{taassori2018vault}. This raises the natural question of whether storage-level hash trees observe similar costs. Our goal in this paper is to quantify this effect and design optimizations to reduce overheads if so.

\begin{figure}[!t]
    \centering
    \includegraphics[width=0.475\textwidth]{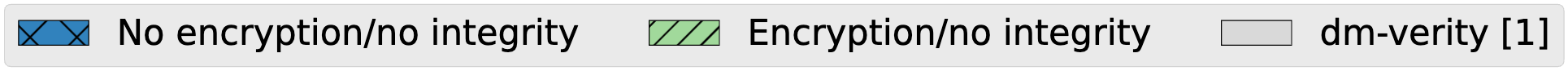}\\
    \includegraphics[width=0.45\textwidth]{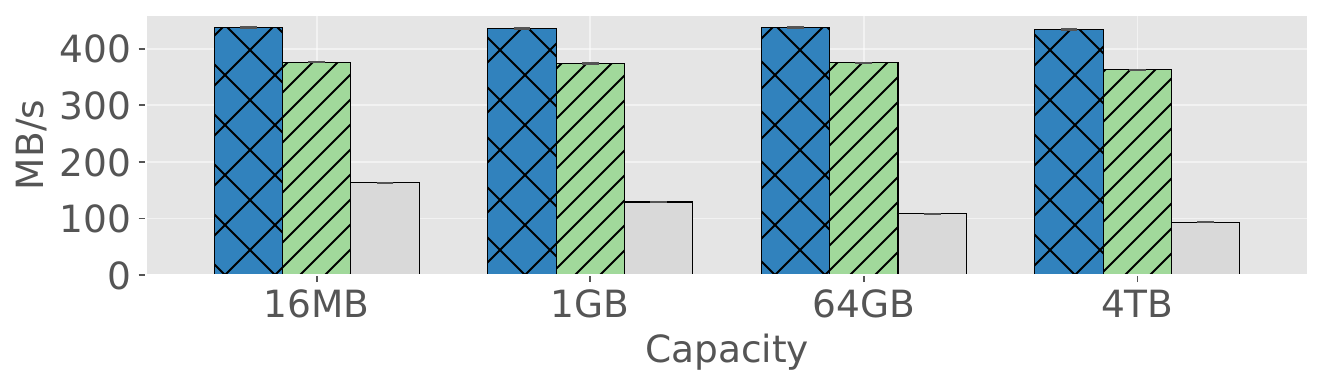}\\
    \caption{This graph shows how throughput decreases w.r.t. capacity under an exemplar setup and workload. Experiment parameters: Workload: Zipf(2.5), Read ratio: 1\%, I/O size: 32~KB, Cache size: 10\%.}
    \label{fig:bottleneck11}
\end{figure}
\begin{figure}[!t]
    \centering
    \includegraphics[width=0.45\textwidth]{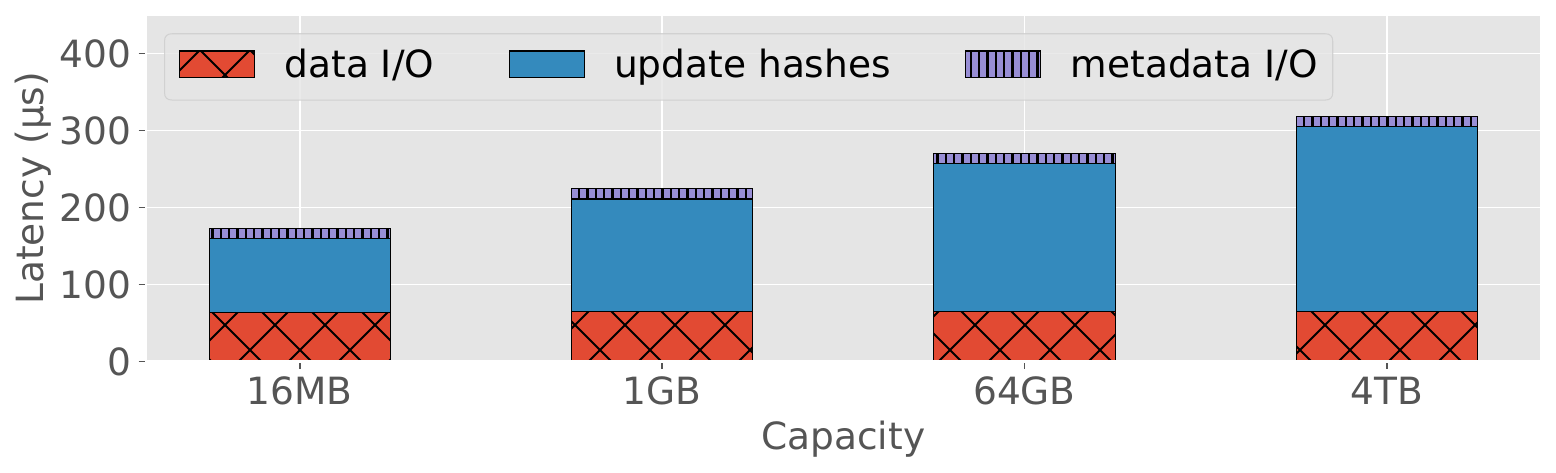}\\
    \caption{CPU vs. I/O time during the driver \texttt{write} routine. Same experiment parameters as above.}
    \label{fig:bottleneck12}
\end{figure}

\shortsection{Scalability Problem}
We begin with a motivating experiment in~\autoref{fig:bottleneck11}, which demonstrates the performance of the state-of-the-art hash tree design used by dm-verity---a balanced, binary tree. The graph shows how throughput changes as disk capacity increases. We defer implementation and experiment setup details to~\autoref{sec:exp-setup}, but note that the hash tree is implemented in a block device driver that wraps a lower-level driver, and is exposed as a regular device to file systems or other applications as \texttt{/dev/XXX}.

The graph shows that throughput decreases w.r.t. capacity. This is due to the tree size (height) increasing logarithmically with capacity, which is reflected in logarithmically increasing slowdowns. At 16~MB capacity, the hash tree incurs nearly a 60\% throughput loss over the Encryption/no integrity baseline. At 4~TB capacity, throughput loss increases to 75\%. Note that the workload shape is immaterial here; the same overheads are always observed because of the tree structure \newtxt{(i.e., all accesses will require computing the same number of hashes)}. \oldtxt{Further, note that read-heavy workloads are generally not an issue,}

\newtxt{Read-heavy workloads do not pose significant challenges,} because the (small) hash cache is very efficient (hit rate >99\%), and verifies benefit from early exits when they hit a cached hash. The problem is how to efficiently handle writes: under write-heavy workloads, hash tree overheads are prohibitive at small and large capacities, undermining the performance capability of the fast NVMe device\footnote{\newtxt{Note that in some cases the page cache can hide disk-level performance impacts, but the performance of large-scale applications will depend heavily on the underlying disk performance (due to writeback contention and throttling under heavy memory pressure~\cite{corbet2010block}). Our focus is thus on analyzing and optimizing the storage layer to ensure that the performance capabilities of fast NVMe devices are not undermined by hash tree operations.}}.

\shortsection{Root Cause Analysis}
\autoref{fig:bottleneck12} shows the latency breakdown during the device driver \texttt{write} routine. As expected, for a 32~KB I/O the time spent pushing data out to disk (data I/O) is approximately $60~\mu s$. The remaining time in the write routine is spent fetching/writing hashes to disk (metadata I/O) and performing hash updates (computing the new block hash and executing the hash tree update). Metadata I/O is negligible because the hash cache is very efficient. The majority of time is therefore attributed to managing the hash tree\footnote{\newtxt{Note that our focus is on NVMe SSDs; HDDs have a different performance profile that we are not optimizing for. With HDDs, the data access time dominates hashing time, so hash tree overheads can be relatively negligible.}}.

To understand why this occurs, \autoref{fig:bottleneck22} shows the latency to compute a SHA256 hash (the standard hash function used in Merkle hash trees) vs. data size on a 2.9GHz Intel Xeon Platinum 8375C, a 3rd Generation Intel Xeon Scalable processor supporting AES and SHA instruction-set extensions to accelerate cryptographic operations. We observe that it takes approximately $490~ns$ to compute the hash of 64~B of data. We also measure the latency to encrypt and generate the MAC for a 4~KB block with AES GCM to be approximately $2~\mu s$.

\begin{figure}[!t]
    \centering
    \includegraphics[width=0.45\textwidth]{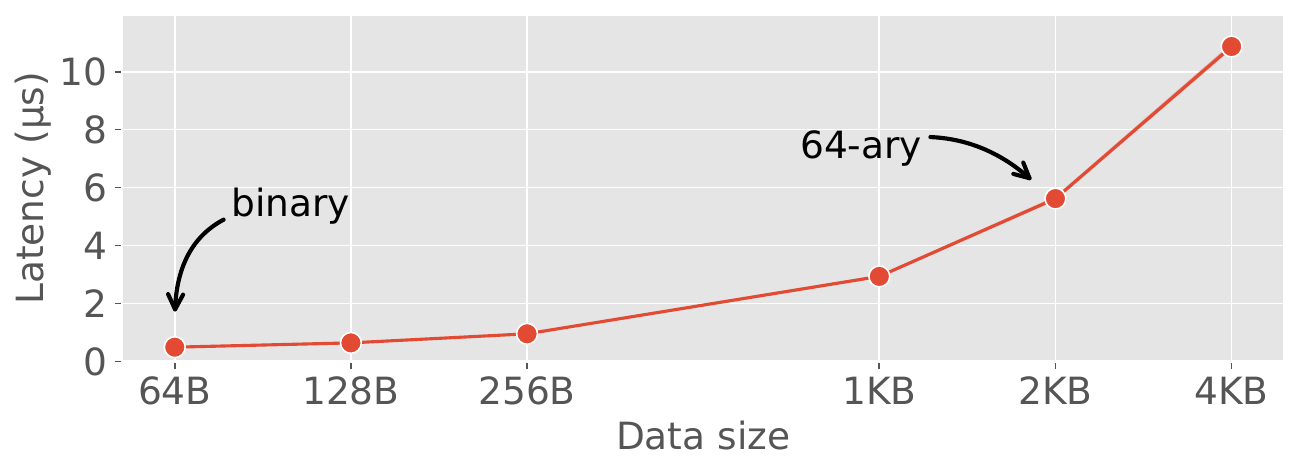}\\
    \caption{This graph shows the latency of computing SHA256 hashes on a modern processor with hardware acceleration for cryptographic functions. The annotations highlight the input data size to the hash function at different tree arities.}
    \label{fig:bottleneck22}
\end{figure}

\autoref{fig:bottleneck12} shows that at 1~GB capacity, approximately $150~\mu s$ is spent managing the hash tree. Consider that a 1~GB disk has 262,144 4~KB blocks and thus a height of 18, requiring one SHA256 computation per level. Further, with 4~KB disk blocks, executing a 32~KB write I/O would require $32768/4096=8$ hash tree updates executed sequentially---best-known methods still rely on a global tree lock to serialize tree updates. This amounts to $150~\mu s/8=18.75~\mu s$ spent encrypting data, generating the MAC, and updating the hash tree. Thus, we have $18.75-2=16.75~\mu s$ time spent doing the actual hash tree update, and $16.75/18=0.93~\mu s$ total time spent doing work at each level in the tree. Most of this time is spent computing the node hash, with the remaining work being cache lookups and buffer copying.

Note that with even faster devices in the future (with single-digit microsecond access latencies), 
the proportion of time spent hashing vs. doing data I/O will grow substantially.

\shortsection{Optimized Tree Structures}
Fundamentally, this means that time spent hashing (\textit{CPU costs}) is the bottleneck. \autoref{sec:perf-eval} will show that caching and parallelization only help to an extent. What is needed is a structurally more efficient tree. 

Prior works optimizing hash trees for memory have largely converged on the idea that high-degree (e.g., 64-ary) trees are the solution to eliminate overheads~\cite{taassori2018vault}. The intuition is that by increasing tree fanout, one can decrease tree height and thus the number of hashes that must be computed per read or write. However,~\autoref{fig:bottleneck23} shows that high-degree trees are actually a suboptimal design choice. We compute the expected hashing costs based on the hashing latency and the height of the tree observed under a given arity for 1~GB capacity (e.g., 64-ary trees have height 3). \newtxt{While tree height is decreased, the graph shows that increased fanout results in high-degree trees incurring the highest expected hashing costs due to hashing more content, and ultimately lower performance.} \oldtxt{The graph shows that high-degree trees incur the highest hashing costs.} We will demonstrate in~\autoref{sec:perf-eval} that high-degree trees fail to reliably scale, and binary trees perform best. We therefore seek a better way to maneuver binary trees to reduce overheads.

\begin{figure}[!t]
    \centering
    \includegraphics[width=0.45\textwidth]{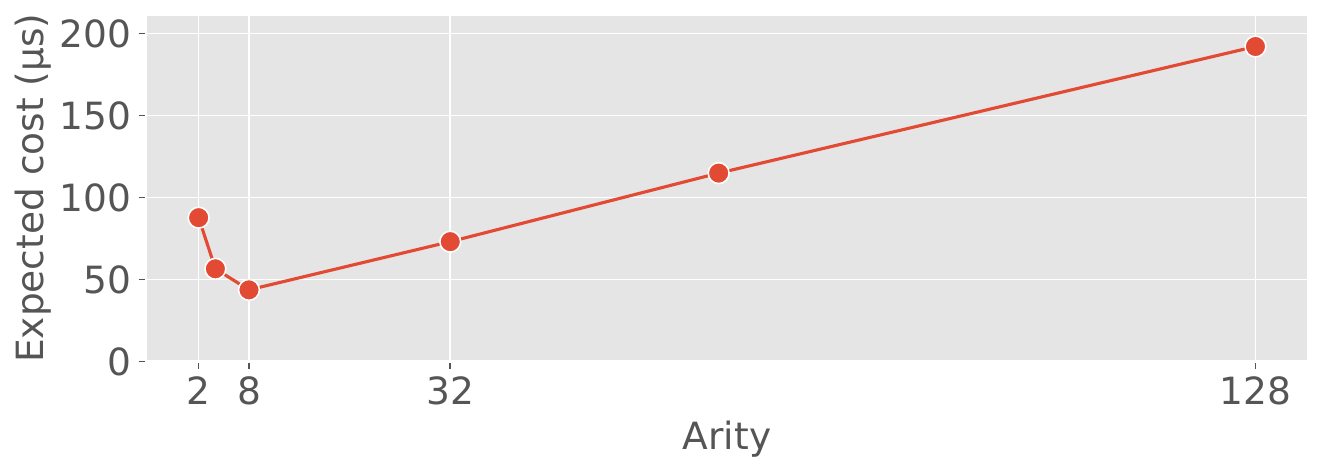}\\
    \caption{We calculate the expected hashing costs for a 32~KB write I/O based on the tree height under different tree arities, given the measured SHA256 latencies for each arity in~\autoref{fig:bottleneck22}. The graph shows that low-degree trees should have lower hashing costs than high-degree trees.}
    \label{fig:bottleneck23}
\end{figure}

\section{Optimal Hash Trees}
\label{sec:optimal}

We approach this problem by asking the fundamental question: \textit{Is there an optimal tree structure?}~\cite{keeton2007don}. Having a definition of an optimal tree serves two purposes: (1) under a specified set of assumptions, it establishes an upper bound on performance, and (2) it discloses what characteristics of the tree structure are correlated with optimal performance.

\subsection{Optimal Definition}
We previously showed that CPU costs are the bottleneck that affect device performance. Intuitively, an optimal hash tree must therefore be a tree that reduces the number of hashes that must be computed per update or verification, reducing hashing costs and therein improving performance. 

We observe that the problem of finding an optimal hash tree can be reduced to finding an optimal prefix tree (or \textit{prefix code}) in the context of lossless data compression~\cite{moffat2019huffman}. Prefix codes map a set of symbols onto a set of codewords, with the goal of compression being that codewords are as short as possible to produce a maximally compressed representation of the original data. An example is shown in~\autoref{fig:huffman}. Formally:

\begin{theorem}
    A hash tree constructed as an optimal prefix code is optimal for an i.i.d. access probability distribution.
\end{theorem}

\begin{proof}
Let $A=\{a_1, a_2, \ldots, a_n\}$ be a symbol alphabet and $W=\{w_1, w_2, \ldots, w_n\}$ be a set of associated symbol weights. Let $C=\{c_1, c_2, \ldots, c_n\}$ be a prefix code that represents the set of codewords for symbols in $A$. A prefix code $C$ is said to be optimal if it minimizes the \textit{expected} codeword length: $\argmin_{C}~\sum_{i=1}^n w_i |c_i|, c_i\in C$. The length of a codeword is the number of bits in a codeword, or equivalently, the \textit{number of edges} in the path from the root to the symbol leaf in the prefix tree representation of $C$. Huffman coding is a widely-used algorithm to produce optimal prefix codes~\cite{huffman1952method}.

Now let $B=\{b_1, b_2, \ldots, b_n\}$ be a set of disk blocks and $F=\{f_1, f_2, \ldots, f_n\}$ be a set of access frequencies to blocks determined by some known workload profile. Suppose we map each block $b_i$ to a symbol $a_i$ and each access frequency $f_i$ to a symbol weight $w_i$. Running Huffman's algorithm on $A$ and $W$ produces a prefix code with expected codeword length $\sum_{i=1}^n w_i |c_i|$~=~$\sum_{i=1}^n f_i |b_i|$.

In the compression domain, the number of edges represents the number of bits needed to parse a symbol's codeword, while in the hash tree domain it represents the \textit{number of hashes} that must be computed from leaf to root for a block. A hash tree constructed as a Huffman code minimizes the expected number of hashes computed during an update or verification and is therefore an optimal hash tree.
\end{proof}

\begin{figure}[t]
    \centering
    \includegraphics[width=0.8\linewidth]{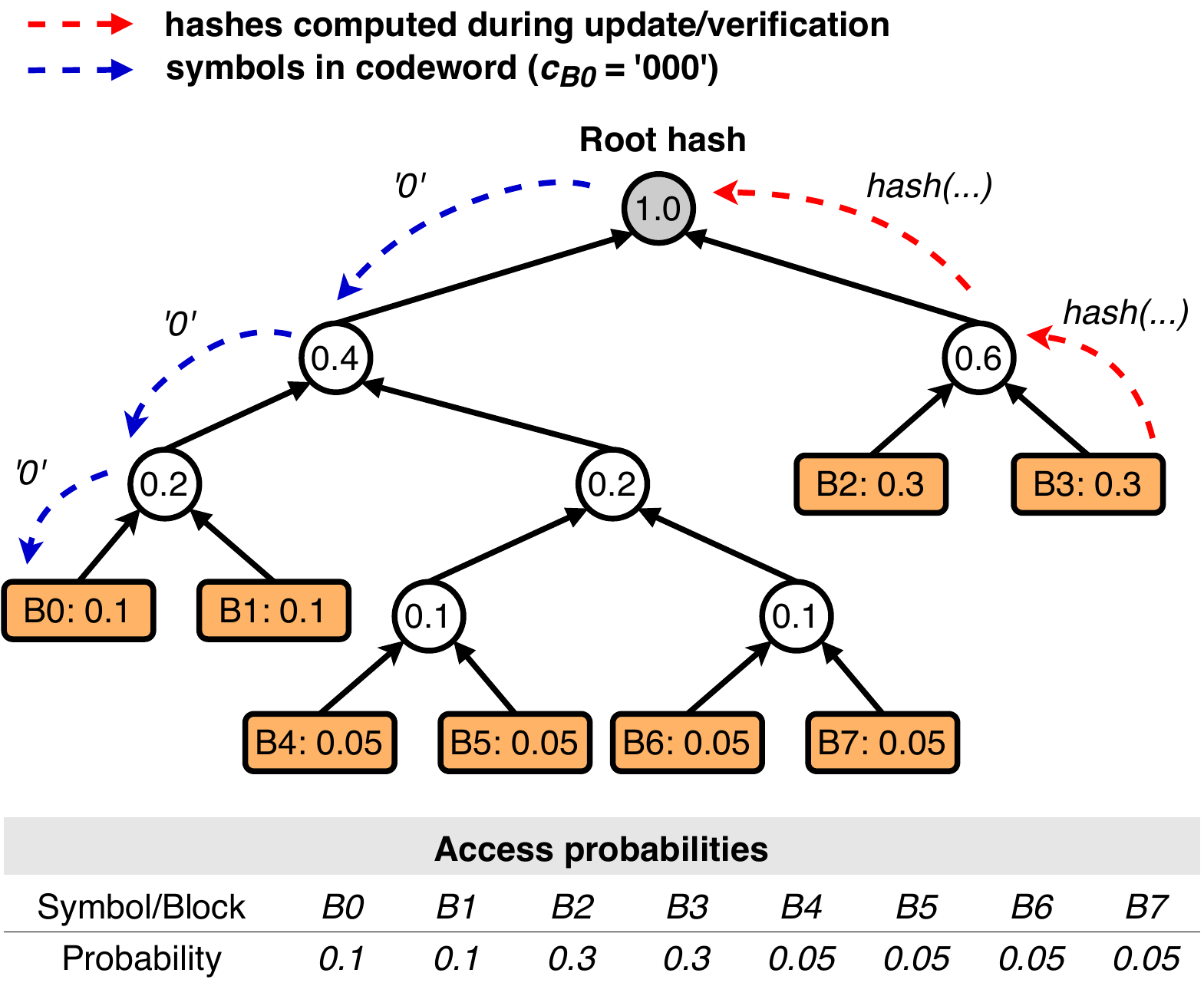}
    \caption{A Huffman tree is an optimal prefix tree. A hash tree constructed as a Huffman tree with a given access probability distribution is an optimal hash tree.}
    \label{fig:huffman}
\end{figure}

\subsection{Extended Optimal Definition}
\label{sec:extended-opt}
Now we extend our optimal definition to consider the effects of cache performance. Note that both data blocks and hashes can be cached in memory. In the compression domain, codeword paths in an (optimal) prefix tree can be parsed in constant work per edge, giving a total amount of work: $\sum_{i=1}^n w_i (|c_i|\cdot O(1))=O(1)\cdot \sum_{i=1}^n w_i |c_i|$. However computing a hash requires (at least) two hash fetches in a binary hash tree: the node's two children. If both hashes are present in memory, fetch costs are negligible and the amount of work is similarly optimal: $\sum_{i=1}^n f_i (|b_i|\cdot O(1))=O(1)\cdot \sum_{i=1}^n f_i |b_i|$. Otherwise if they must be fetched from disk, I/O costs are non-negligible: $\sum_{i=1}^n f_i (|b_i|\cdot t(b_i))=\sum_{i=1}^n f_i |b_i|\cdot t(b_i)$, for some function $t(b_i)$. 

We can model the incurred I/O costs using the average memory access time formula:
\begin{equation}
    \begin{aligned}
    \text{AMAT} & = \text{hit time} + \text{miss rate} \times \text{miss penalty} \\
    \implies t(b_i) & = \text{mem latency} + \text{miss rate} \times \text{reauth latency} \\
    \implies t(b_i) & = H + mD = O(1) + mD
    \end{aligned}
\end{equation}

where $H$ is a fixed memory access cost, $m$ is the miss rate of a node fetch in memory, and $D$ is a fixed fetch/reauthentication cost. Substituting this in, the total amount of work is:

\begin{equation}
    \begin{aligned}
    \sum_{i=1}^n f_i |b_i|\cdot t(b_i) & = \sum_{i=1}^n f_i |b_i|\cdot (O(1) + mD) \\    
    & = \underbrace{O(1)\cdot\sum_{i=1}^n f_i |b_i|}_{\text{base work}} + \underbrace{mD\cdot\sum_{i=1}^n f_i |b_i|}_{\text{I/O costs}}. \\
    \end{aligned}
\end{equation}

\shortsection{Remark}
From our model, we see that higher miss rates for block hashes incur more work per edge, proportional to the expected number of hashes that must be computed per update or verification. Specifically, at a given miss rate, the incurred I/O costs follow the same distribution as the underlying access probability distribution: hotter data has a lower expected amount of base work and incurs lower I/O costs, while colder data has a higher expected amount and incurs higher I/O costs.

We also see that with an optimal cache ($m=0.0$), the expected total amount of work is exactly optimal. However, it has been empirically observed that as cache size decreases, miss rates increase with a power law~\cite{chow1974optimization}, and thus as the cache size decreases, expected I/O costs increase with a power law. This implies that the performance of hash trees is very sensitive to cache size. In particular, as cache memory can be financially costly on cloud servers, being able to synergize well with relatively smaller caches (w.r.t. larger disks) is critical to a practical hash tree deployment.

Moreover, a Huffman tree is optimal under a \textit{known} and \textit{fixed} set of weights (cf. access probability distribution) for an \textit{i.i.d.} source. If the symbol sequence (cf. block access sequence) observed while compressing a message (cf. during a program trace) exactly matches the one used to construct the tree, then the tree will be exactly optimal (i.e., provide optimal throughput). However, if the sequence deviates from the one used to construct the tree, the tree will not be exactly optimal. Similarly, if the source is not \textit{i.i.d.}, then the tree will not be exactly optimal---temporal patterns in the workload may cause the tree to underestimate an upper bound.

\subsection{Optimal Tree Oracle}
\label{sec:oracle}

Our optimal definition provides that, if we have knowledge of a concrete block access sequence (i.e., workload trace), we can instantiate an optimal hash tree from the trace and measure a concrete upper bound on performance (i.e., maximum possible throughput under the given workload). In an \textit{offline} setting, where we have access to workload traces (e.g., recorded with tools like \textit{blktrace} or \textit{fio}), we can feasibly do so. We refer to this methodology as the optimal tree oracle. 

The primary purpose is to measure whether overheads observed by a hash tree design (like dm-verity) are better attributed to the structure of the tree or to a fundamental scaling limit. For example, a hash tree may be performing optimally, but have high overheads, which would suggest that complimentary optimizations (e.g., dividing the tree into one or more independent security domains) may be the only way to break the performance ceiling. In contrast, a hash tree that does not perform optimally under a given workload may require a fundamental redesign.

We liken this approach to Belady's optimal page replacement algorithm~\cite{belady1966study}, a clairvoyant algorithm that has a priori knowledge of future memory accesses and can make optimal page replacement decisions. This gives us the ability to make rigorously grounded conclusions about what hash tree designs perform well and when. We defer analysis with the optimal tree oracle (denoted by H-OPT) to~\autoref{sec:perf-eval}.

\section{Dynamic Merkle Trees}
\label{sec:dmt}

A condition of instantiating an optimal hash tree is that we must have a priori knowledge of the exact workload. This is rarely feasible in practice. This section builds on observations of optimal hash trees to develop a novel hash tree design that can approximate an optimal tree by learning and adapting to workload patterns on-the-fly.

\begin{figure}[t]
    \centering
    \texttt{zipf2.5}
    \includegraphics[width=0.45\textwidth]{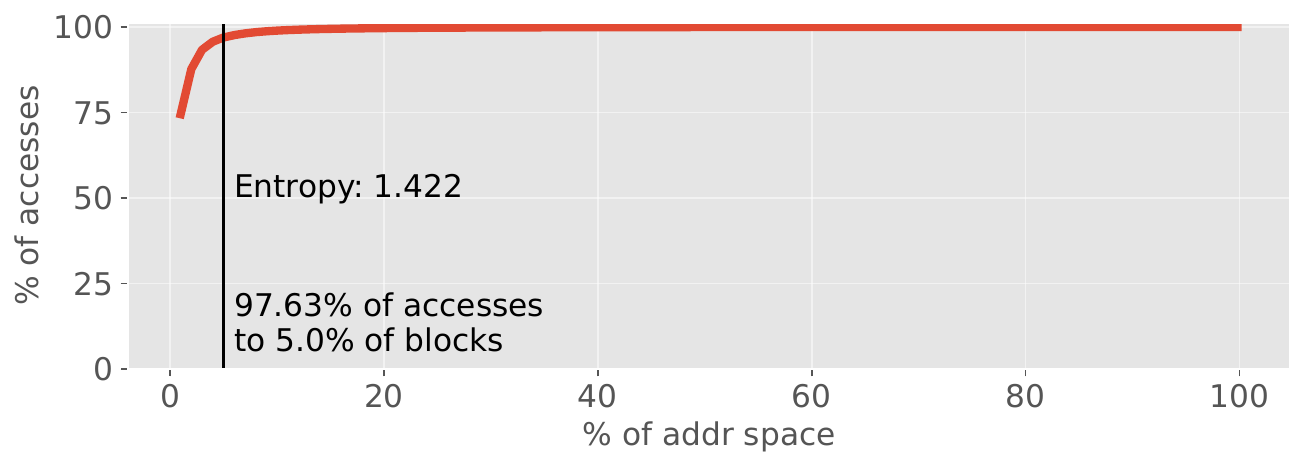}\\
    ~\\
    \caption{In practice, access patterns observed at the storage layer are most often skewed, and Zipfian workloads are used to model this shape~\cite{li2023depth,yang2016write}. This graph shows that a small number of blocks is accessed most of the time, which suggests that operations on the hash tree should also be skewed.}
    \label{fig:skewed}
\end{figure}

\subsection{Challenges}
\shortsection{Finding a Suitable Tree Structure}
State-of-the-art hash tree designs rely on static, balanced tree structures~\cite{android-dm-verity,taassori2018vault}. Balanced trees are optimal under uniform access patterns. However, real-world storage workloads most often exhibit skewed (i.e., non-uniform) access patterns~\cite{li2023depth,wang2023locality,cooper2010benchmarking,yang2016write,arasu2021fastver} where a small number of blocks is accessed most of the time. The result is that the optimal hash trees produced by Huffman codes are often far from balanced.

For example,~\autoref{fig:skewed} shows the access distribution for a Zipfian workload; note that most real-world workloads obey Zipf's law and are modeled with Zipfian workloads~\cite{yang2016write}. The data accesses are highly skewed, which suggests that the operations on the hash tree may also be highly skewed. \autoref{fig:unbalanced} shows this to be true: in a balanced tree over 8192 blocks (a 32~MB disk), leaf node heights are constant at 13, but in the optimal tree we see two distinct regions representing hotter (height~$\approx 10$) and colder data (height~$\approx 30$). 

Optimal trees tend to accumulate hot data high up in the tree and place cold data at nearly a $3\times$ height difference---significantly reducing the verify/update latency for hot data. This indicates that an optimal tree is one that aggressively optimizes for hot data (the working set). Unfortunately, the \textit{static} nature of standard hash tree designs precludes exploiting this skew when present in a workload.

\shortsection{Handling Changing Access Patterns}
Yet, workload characteristics can also vary over time: access patterns may still be skewed but regions of interest may change, or some periods of time may be characterized by more uniform access patterns. This is particularly true for storage that is shared by multiple cooperating applications or users. Thus, a tree that is optimal at one point in time may not be optimal for another (i.e., dynamically optimal). An online solution, one that does not assume a priori knowledge of workload characteristics, therefore must not only be able to \textit{capture} hot data by placing more frequently accessed nodes higher in the tree, but also be able to dynamically \textit{adapt} to changes in what particular data is deemed hot or cold over time.

Adaptive tree structures have been widely studied, particularly for search. Most algorithms focus on keeping trees balanced to reduce worst-case running time. We explicitly aim to remove this constraint; commonly used self-balancing trees (e.g., AVL trees) are therefore ill-fit for our use case. We aim for a more aggressive optimization: allow the tree to become unbalanced as necessary, but be driven by the workload.

\begin{figure}[t]
    \centering
    \includegraphics[width=.65\linewidth]{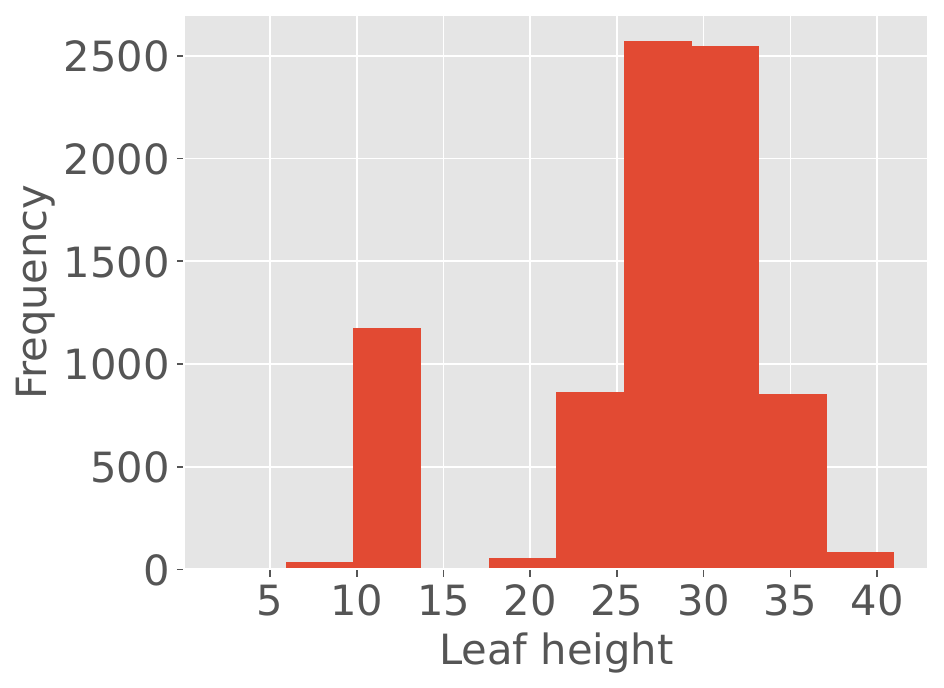}
    \caption{Skewed patterns manifest in optimal hash trees being far from balanced: under a Zipfian probability distribution, there are distinct regions of relatively hotter and colder data.}
    \label{fig:unbalanced}
\end{figure}

\shortsection{Managing Restructuring Costs}
While intuitively it makes sense that unbalanced trees should be able to exploit skewed patterns by placing more frequently accessed leaf hashes closer to the root, realizing this is in a real system is a non-trivial problem. The mechanics of adapting trees involve a series of rotations. While rotations are cheap for search trees, consisting of a series of pointer updates, they are expensive for hash trees, as we have to recompute hashes for all nodes from the rotation point up to the root. This applies when nodes are rotated during either verifications or updates.

The costs of rotating nodes in the tree may therefore quickly outweigh any expected benefits of moving frequent nodes closer to the root. The cost of a rotation itself is also not constant, but proportional to the current height of the nodes involved in the rotation. Further, search trees permit all nodes to be searchable, but as mentioned, only leaf nodes are searchable in a hash tree. We therefore must maintain the invariant that during a rotation, a leaf remains a leaf and an internal node remains an internal node. Otherwise, a rotation will result in an invalid tree structure. 

\subsection{Randomized Splaying}
We draw a connection to a data structure widely used in garbage collection and IP routing: splay trees~\cite{sleator1985self}. Splay trees are a type of binary search tree that brings an accessed leaf to the root through a series of rotations. Importantly, splay trees capture temporal locality by keeping frequently accessed nodes closer to the root (\autoref{fig:splay}). 
However, naively used, the cost of splaying can be extremely expensive, and splaying too frequently or opportunistically may keep the tree more balanced than desired. We adapt the conventional splay tree design to meet the constraints of a hash tree.

\shortsection{Heuristic Parameters} We define three parameters: a splay window flag $w$, splay probability $p$, and splay distance $d$. The splay window flag can be toggled on or off to indicate whether or not the splay window is active (i.e., whether or not we should consider a tree node to be splayed). This notion is useful because there may be certain periods at runtime where splaying should necessarily not occur. This may be the case, for example, if the system administrator has knowledge of current application access patterns or profiles them periodically, or if other background storage tasks may be in progress that require stability of data (e.g., health checks).

If the splay window is active, the splay probability denotes the probability that an accessed node should be splayed. The key intuition is that splaying is an expensive operation, but we can amortize costs by only splaying on a small percentage of accesses (e.g., 1\% of the time). Finally, if a node is decided to be splayed, the splay distance defines the maximum number of levels that the node should be splayed (i.e., \textit{promoted}) up the tree. This entire process occurs at the end of each verification or update call and before anything is returned to the caller. 

\subsection{Technical Approach}
\shortsection{Analyzing Data Hotness}
The splay distance is the central parameter that determines the effectiveness of a splay operation. Determining a suitable splay distance is challenging, as there is an inherent risk vs. reward trade-off when splaying a node. Splaying any node all the way to the root may be very beneficial if a node is relatively hot, as future accesses to the data can quickly benefit from the promotion. However, doing so would be severely wasteful if the node is cold, as the tree will then require several \textit{additional} rotations to eventually promote hotter data and demote the cold data from a higher position in the tree. Finding an accurate and practical hotness metric is critical to balancing this trade-off.

We attach an integer hotness counter to each tree node that is incremented whenever a node is promoted in the tree and decremented whenever a node is demoted in the tree. This applies to both leaves (i.e., blocks) and internal nodes (which indicate the hotness of particular subtrees/blocks). The purpose of the hotness counter is to track the relative access frequencies of nodes as they are rotated, and use this information to determine how far to splay the node up the tree.

\begin{figure}[t]
    \centering
    \includegraphics[width=0.9\linewidth]{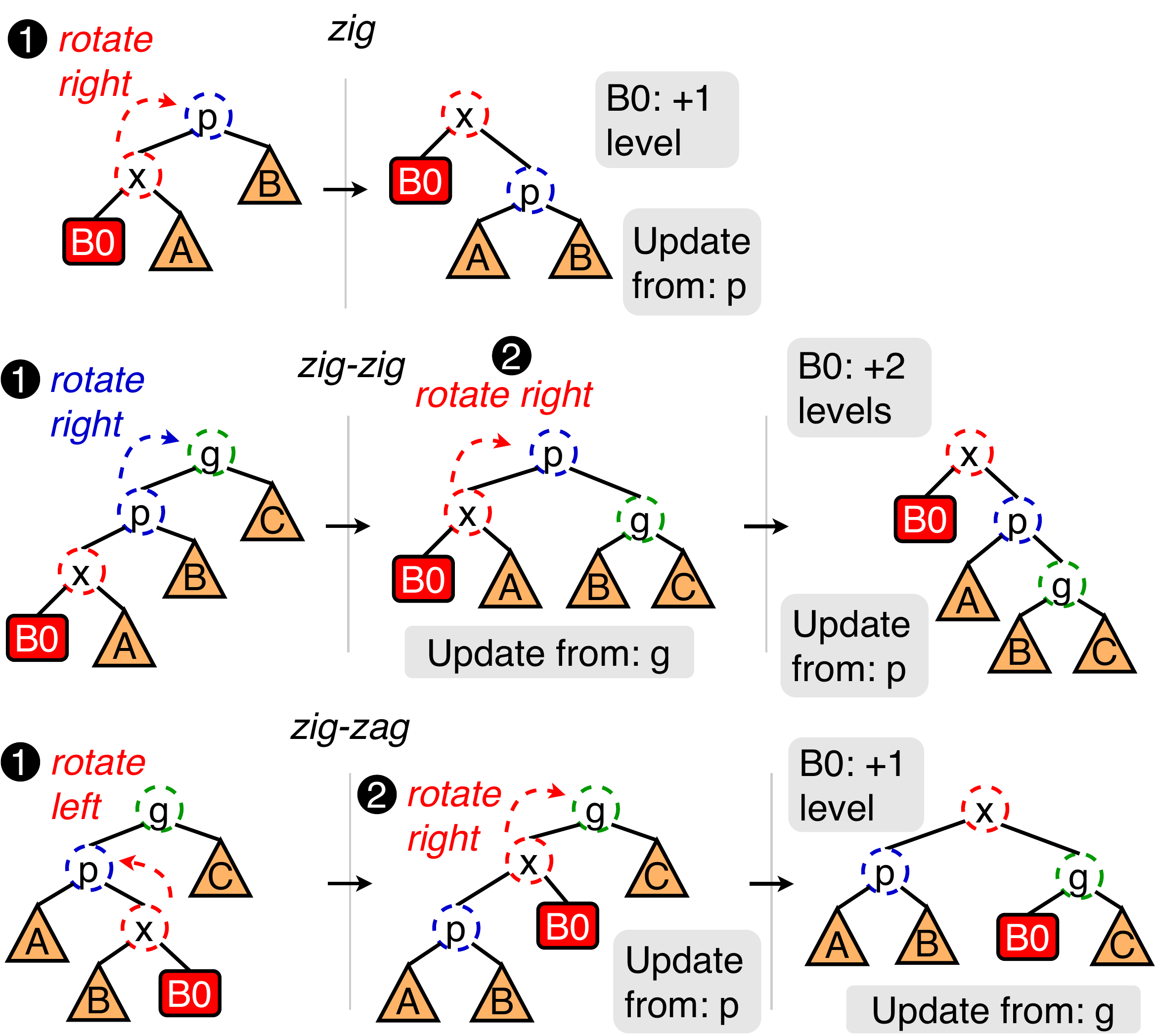}
    \caption{Splay trees are a type of binary search tree that capture temporal locality by bringing an accessed leaf closer to the root. We use a splay-based hash tree design to similarly capture temporal locality in cloud block storage workloads.}
    \label{fig:splay}
\end{figure}

The counter is initialized to zero after the node is authenticated and cached; the hotness of nodes that are not currently cached in memory is therefore not tracked. 
The purpose of this is to localize our analysis of data hotness to the working set. Note that this approach can negatively affect performance for small caches, as it will be difficult to draw a contrast between relatively hotter, warmer, and colder data when counters are reset frequently. Nonetheless, the splay distance is a function of the hotness. The splay distance is computed in a straightforward manner: at a distance proportional to the hotness. For simplicity, we set the splay distance to be $h$ levels, where $h$ is the current hotness counter value of the accessed leaf.

Intuitively, nodes that are deeper in the tree (colder) will climb the tree slowly, while nodes that are higher in the tree (hotter) will climb quickly. Note that our initial exploration into this space could be expanded with sketching algorithms, machine learning, or other sophisticated techniques~\cite{hashemi2018learning}.

\shortsection{Promotion \& Demotion}
After computing the splay distance, the final step is to execute the splay operation. Splaying a DMT is done in nearly the same way as it is in a search tree. There are three cases to consider when splaying a node: \textit{zig}, \textit{zig-zig}, and \textit{zig-zag} (\autoref{fig:splay}). In a zig case, the node's parent is the root, and we rotate the node up to the root. In a zig-zig case, the node's parent is not the root, and the node and the node's parent are either both left or right children. In this case, we perform two rotations along the same direction to rotate the node up two levels. In the zig-zag case, the node's parent is not the root, and the node and the node's parent are opposite-side children. Thus, we perform two rotations along opposite directions to rotate the node up two levels.

A consequence of splaying is that an accessed node will either be promoted two levels (or to the root). Neighboring nodes will similarly be promoted opportunistically as a side-effect of the splay. This provides two benefits. Nodes that are accessed frequently will therefore have an increasingly shorter path to the root, making verifications and updates quicker. More subtly, nodes that are accessed in close temporal proximity will slowly accumulate in nearby regions of the tree, allowing to exploit spatial locality within the tree.

\shortsection{Maintaining Hash Tree Invariants}
We make three key changes to the standard splay operation to maintain three tree invariants. First, during a splay, we must ensure that a leaf node remains a leaf node and an internal node remains an internal node. Otherwise, a rotation will result in an invalid tree structure. For example, if a leaf node is splayed to the root, the root will become a leaf node, which is invalid. Whenever a block is read or written, we therefore execute a splay on the accessed leaf's parent rather than the leaf.

Next, we propagate the child status (left/right) to the splay operation, and swap the children of the parent node and the accessed node where necessary.
This preserves the structural constraint for a valid hash tree while still ensuring the maximum degree of promotion for the accessed node.

Finally, splaying naturally introduces inconsistency into the tree, as it alters parent-child relationships. This will cause any subsequent hash fetches on a cache miss to fail due to an inconsistent root hash. We must therefore ensure that the tree remains in a consistent state by preemptively fetching (and authenticating) all sibling hashes before performing a rotation, then using them to commit the change immediately after. That is, parent hashes up to the root are recomputed per rotation (see "Update from" in~\autoref{fig:splay}). While updates can be costly, splaying can reduce these costs over time.

\section{Evaluation}
\label{sec:perf-eval}

We compare DMT performance against two insecure baselines (No encryption/no integrity, Encryption/no integrity), the \oldtxt{state-of-the-art} binary trees used by dm-verity~\cite{android-dm-verity}, and the \newtxt{optimal binary tree} \oldtxt{optimal tree oracle} (H-OPT). We also juxtapose DMTs against the high-degree trees that have been widely used in secure memory systems~\cite{taassori2018vault,feng2021scalable}. Finally, as noted previously (see~\autoref{fig:bottleneck23}), we also examine DMTs with respect to 4-ary and 8-ary trees, which have not been considered by prior work. An overview of our parameters is shown in~\autoref{tab:parameters}.

\subsection{Experiment Setup}
\label{sec:exp-setup}

\begin{table}[t]
    \centering
    \small
    \begin{tabular}{ll}
        \toprule
        \textbf{Parameter} & \textbf{Description} \\
        \midrule
        \textit{Capacity} & Usable capacity for data blocks \\
        \textit{Cache size ratio} & Cache size as \% of tree size \\
        \textit{Read ratio} & \% of read operations \\
        \textit{I/O size} & Size of application I/O \\
        \textit{I/O depth} & Max no. outstanding application I/Os \\
        \textit{Thread count} & Number of application threads \\
        \bottomrule
    \end{tabular}
    \caption{Experiment parameters.}
    \label{tab:parameters}
\end{table}

\shortsection{Implementation}
We implement the hash trees in 5~K lines of C++. We use BDUS to implement custom block device drivers that wrap lower-level drivers~\cite{faria_bdus_2021}. BDUS has a kernel module that exposes block layer hooks to userspace. The two primary functions of interest are \texttt{read()} and \texttt{write()}, which are invoked by the kernel whenever a block is read from or written to the block device, respectively. We perform a verify immediately after a block is read and an update  immediately before a block is written to disk. Our basic data unit aligns with the disk I/O size (4~KB blocks)~\cite{khati2017full,brovz2018practical,chakraborti2017dm,nist-aes-xts}.

\shortsection{Testbed}
We perform all experiments using AWS EC2 i4i.8xlarge instances equipped with 32 cores, 256~GB memory and locally-attached NVMe SSDs for data and metadata. Note that currently there are no available cloud instance types which both have local NVMe storage and also support confidential VM technology such as AMD SEV-SNP. We reinitialize hash trees between each experiment, use a standard LRU cache replacement policy, and we set the splay window flag $w=True$ and splay probability $p=0.01$ for DMTs.

\shortsection{Cryptographic Settings}
Like prior works, we ensure deterministic authenticated encryption with AES-GCM~\cite{avanzi2022cryptographic,taassori2018vault}. We use a 128-bit encryption key for blocks. The MACs produced during the encryption process are used as the leaves in the hash tree. For internal nodes, we compute 256-bit hashes using SHA-256 with a 256-bit key.

\shortsection{Workload Settings}
We perform a broad analysis across different system and workload configurations, parameterized by disk capacity, hash cache size, read/write ratio, I/O size, thread count, and I/O depth. This enables exhaustively examining the performance space of DMTs. We also examine different degrees of workload skewness, from pure uniform to highly skewed. We focus especially on the Zipfian workload discussed in~\autoref{sec:motivation}, which closely approximates real-world block-level access patterns, which are highly skewed (and write-heavy)~\cite{li2023depth,yang2016write}; see~\autoref{fig:dist}. We also use a recently published Alibaba dataset recorded from an array of 1000 volumes backing various virtual machines in a public cloud datacenter~\cite{li2023depth}. Finally, we demonstrate how driver-level improvements translate to application-level improvements with a case study of the Filebench OLTP workload~\cite{Tarasov2016FilebenchAF}.

Like prior works, we generate workloads with fio~\cite{fio}; we record/replay traces for the optimal. Workloads have a 5 minute warmup period and 15 minute benchmark period.

\subsection{Results}
We focus our analysis on three questions:
\begin{enumerate}
    \item \textit{How well do state-of-the-art hash tree designs perform across the various system and workload settings that characterize cloud block storage deployments?}
    \item \textit{To what extent can DMTs improve performance over the state of the art, and under what conditions?}
    \item \textit{What memory and storage trade-offs do DMTs make?}
\end{enumerate}

\begin{figure}[t]
    \centering
    \includegraphics[width=0.5\textwidth]{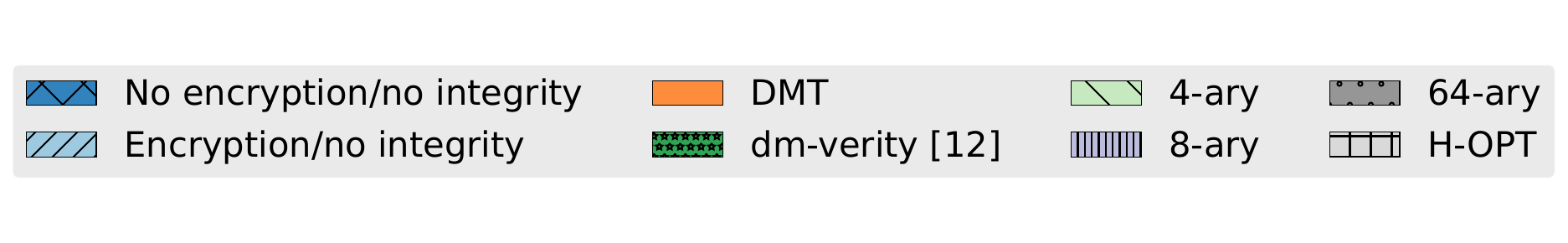}
    \includegraphics[width=0.45\textwidth]{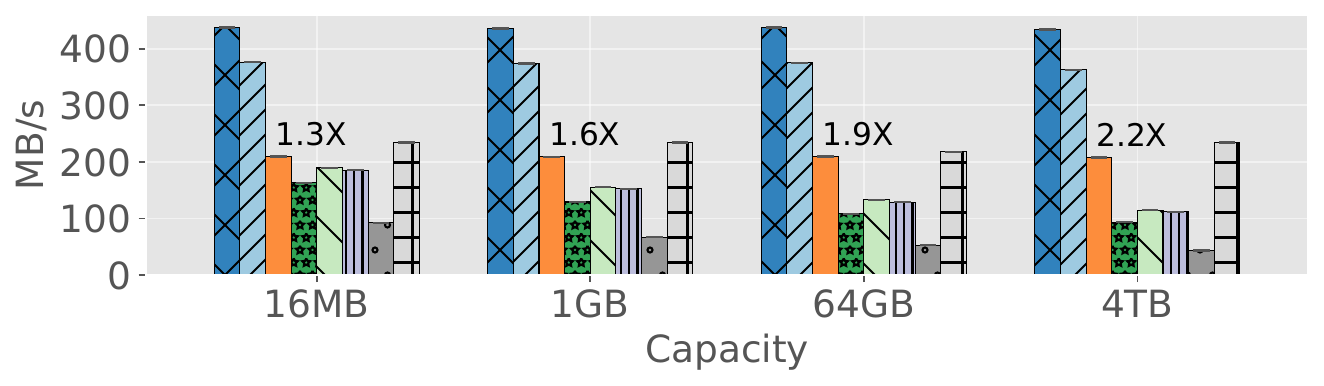}
    \caption{\newtxt{Aggregate throughput with Read ratio at 1\%. DMTs provide up to $2.2\times$ higher throughput than the state-of-the-art, demonstrating its ability to scale well with capacity.} \oldtxt{Aggregate throughput with Read ratio at 1\%. DMTs scale better than the state-of-the-art, providing up to $2.2\times$ throughput improvements.}}
    \label{fig:perf-vs-cap-zipf-tp}
\end{figure}
\begin{figure}[t]
    \centering
    \includegraphics[width=0.45\textwidth]{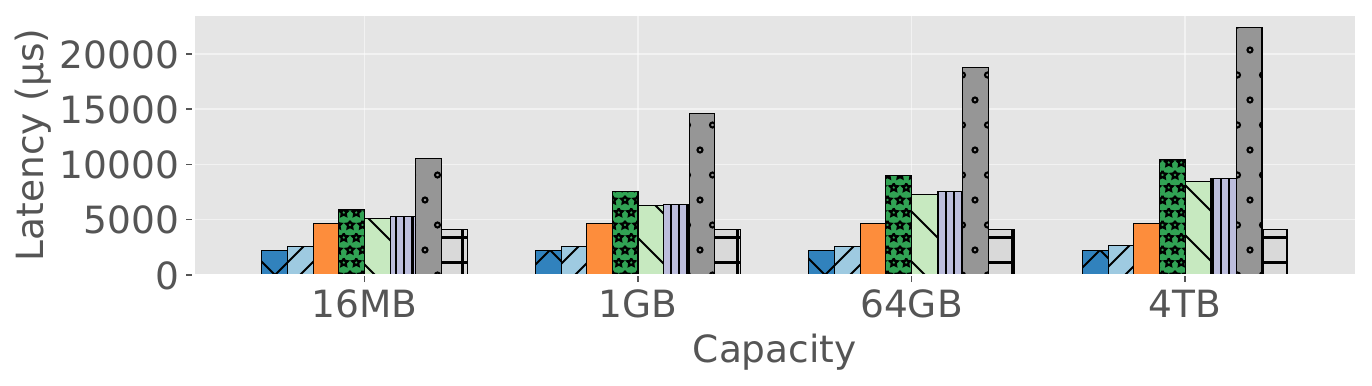}
    \includegraphics[width=0.45\textwidth]{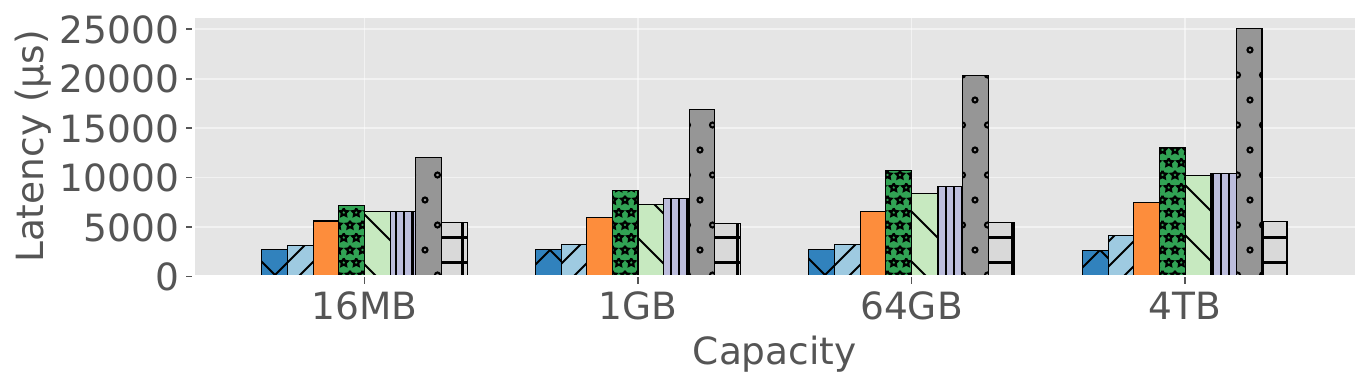}
    \caption{P50 (top) and P99.9 (bottom) write latency. \newtxt{DMTs median and tail latencies reflect throughput improvements, demonstrating it can provide a stable performance guarantee.} \oldtxt{ DMTs show up to 2.2× median and tail latency improvements.}}
    \label{fig:perf-vs-cap-zipf-lat}
\end{figure}

\shortsection{Scaling with Capacity}
We first analyze how disk capacity (which affects tree height) and hash cache size affect performance. Where appropriate, default parameters include---Read ratio: 1\%, I/O size: 32~KB, Thread count: 1, I/O depth: 32, Capacity: 64~GB, Cache size: 10\%. We choose these parameters because they showcase the best performing configuration for the baselines. We examine various workload shapes ranging from uniform to highly skewed Zipfian. We focus particularly on $\theta$: 2.5 because it closely approximates the shape of real-world storage workload patterns~\cite{li2023depth,yang2016write} (see~\autoref{fig:dist}).

\autoref{fig:perf-vs-cap-zipf-tp} shows that aggregate read/write throughput decreases w.r.t. capacity for all balanced trees. Note that the Zipfian workload is emitted from an \textit{i.i.d.} source and is therefore an exact upper bound. We observe that 64-ary trees are the worst performing: reduced tree height can reduce the effective number of hashes that must be computed, but results in lower cache efficiency, which amplifies metadata I/O costs. The state-of-the-art binary trees incur up to a 75\% throughput loss over the Encryption/no integrity baseline at 4~TB. 4-ary and 8-ary trees similarly suffer from a 70\% throughput loss at 4~TB. In contrast, DMTs consistently deliver the highest throughput and >85\% of optimal throughput across all capacities. This clearly demonstrates that DMTs can scale to higher or lower capacities more efficiently. And as noted in~\autoref{sec:motivation}, with even faster storage devices, the proportion of time spent hashing vs. performing the data access will grow substantially, increasing our observed DMT speedups.

Latency improvements are the same---DMTs splay on 1\% of accesses and those costs are amortized over time because splaying occurs most frequently on hot data. \autoref{fig:perf-vs-cap-zipf-lat} corroborates this: DMT median and tail write latencies are still significantly lower than the state of the art.

\begin{figure}[t]
    \centering
    \includegraphics[width=0.45\textwidth]{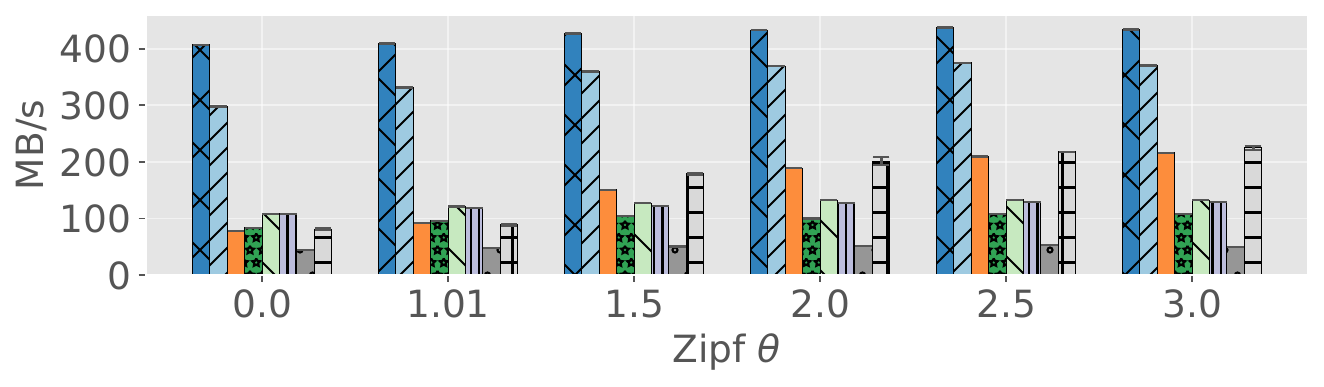}
    \caption{Aggregate read/write throughput. DMTs perform best under skewed workloads; under uniform workloads they observe a 6\% cost over binary trees due to exploratory splays.}
    \label{fig:perf-vs-skew}
\end{figure}

\begin{figure}[t]
    \centering
    \includegraphics[width=0.45\textwidth]{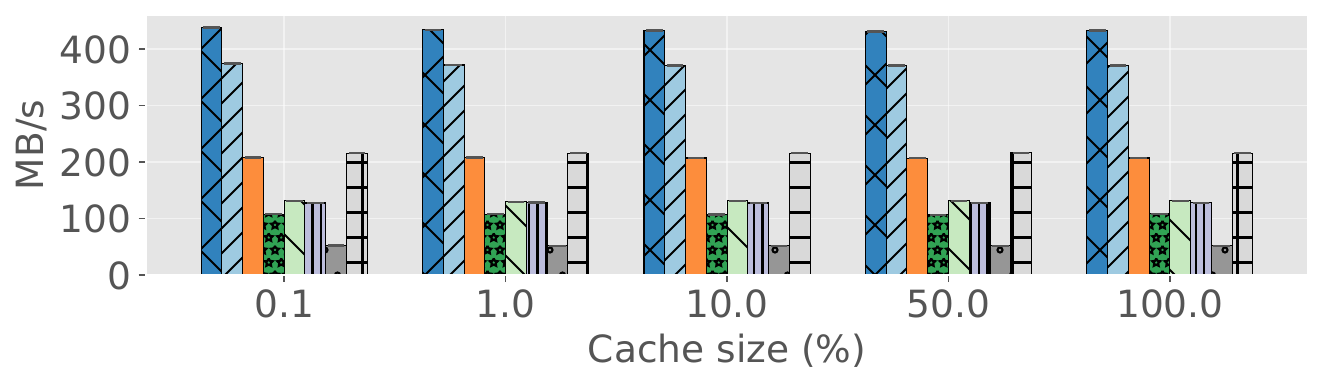}
    \caption{Aggregate throughput. DMTs maintain the highest throughputs across both small and large cache sizes.}
    \label{fig:perf-vs-cache}
\end{figure}

\shortsection{Impact of Workload Skewness}
\autoref{fig:perf-vs-skew} shows how performance changes w.r.t. workload skewness. We observe again that 64-ary trees are the worst performing. DMTs provide up to $2\times$ speedups over the state of the art binary trees under heavy skew, but incur a 6\% cost under more uniform patterns due to exploratory splays which yield no benefit. We attribute this low cost to the fact that DMTs inherit the theoretical guarantees of splay trees, which provide $O(log~n)$ amortized lookup (i.e., verification or update) time. Thus, DMTs perform at least as good as balanced (binary) trees on average. 

We also observe that 4-ary and 8-ary trees deliver 25\% higher throughput than DMTs under more uniform workloads. As discussed in~\autoref{sec:motivation}, low-degree trees hit the optimal points in the design space for \textit{balanced} trees (reduced tree height without adverse effects on cache performance); prior works have not considered this. However, when workloads become skewed, there is a substantial opportunity cost: 4-ary and 8-ary trees do not perform better than the optimal binary tree. This highlights a key observation: increasing tree degree alone is not sufficient to maximize performance. We believe that extending the DMT design to 4-ary and 8-ary trees will yield the most performant and generalized solution.

\begin{figure}[t]
    \centering
    \includegraphics[width=0.5\textwidth]{figures/o3/legend-dmt.pdf}
    \includegraphics[width=0.45\textwidth]{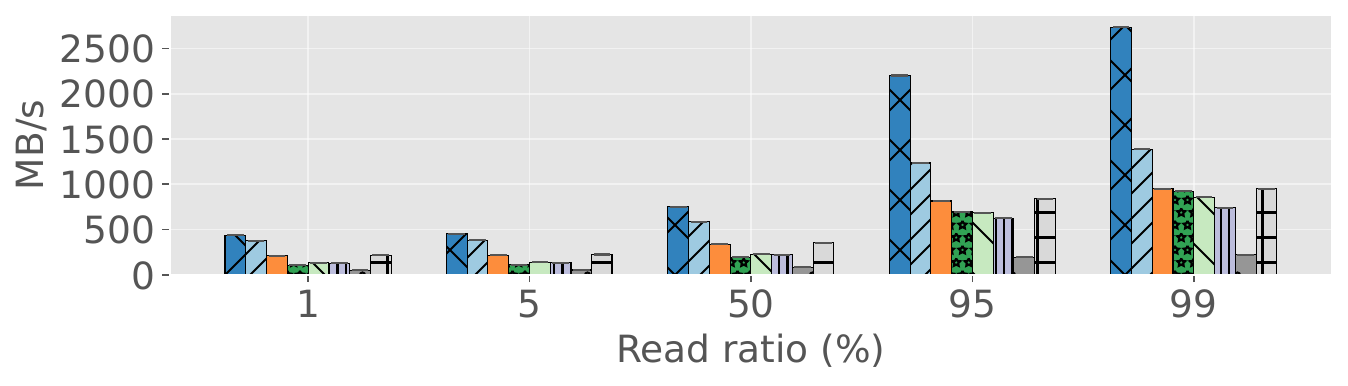}
    \includegraphics[width=0.45\textwidth]{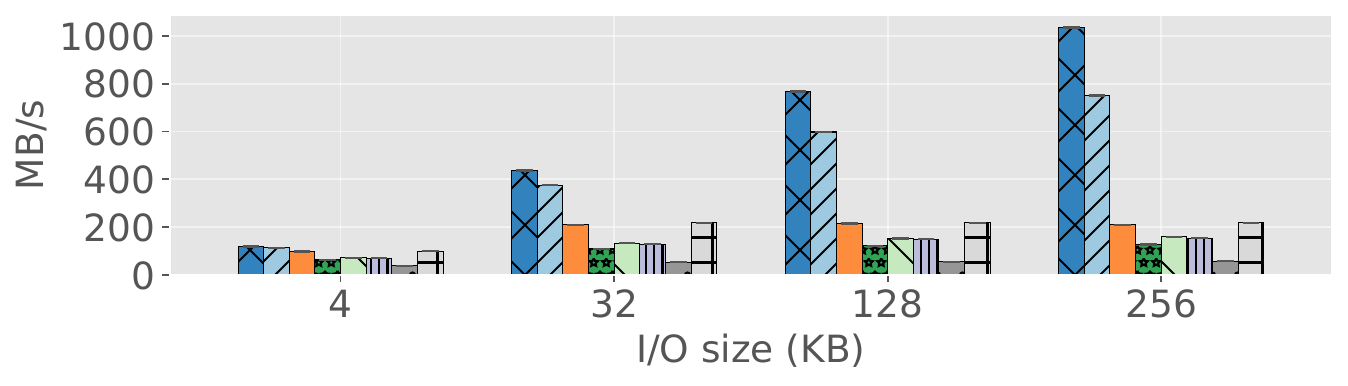}
    \includegraphics[width=0.45\textwidth]{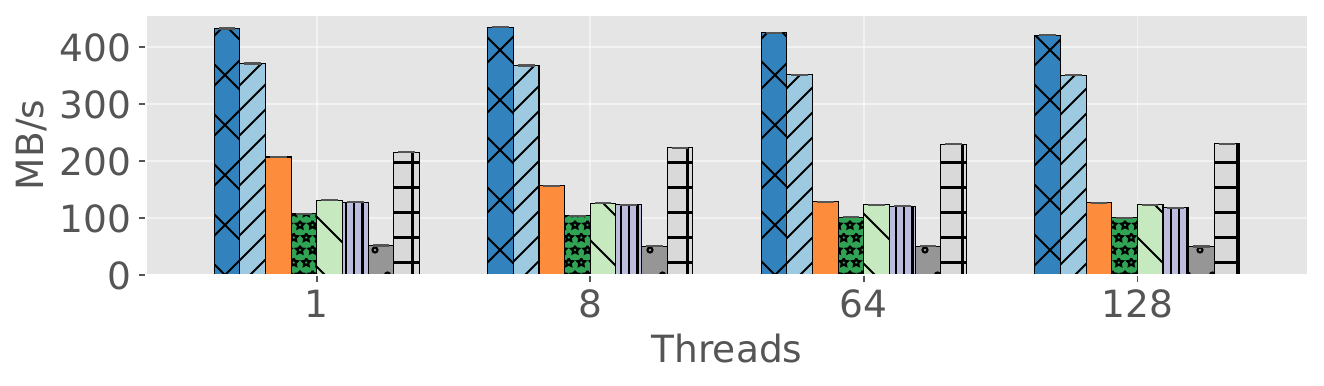}
    \includegraphics[width=0.45\textwidth]{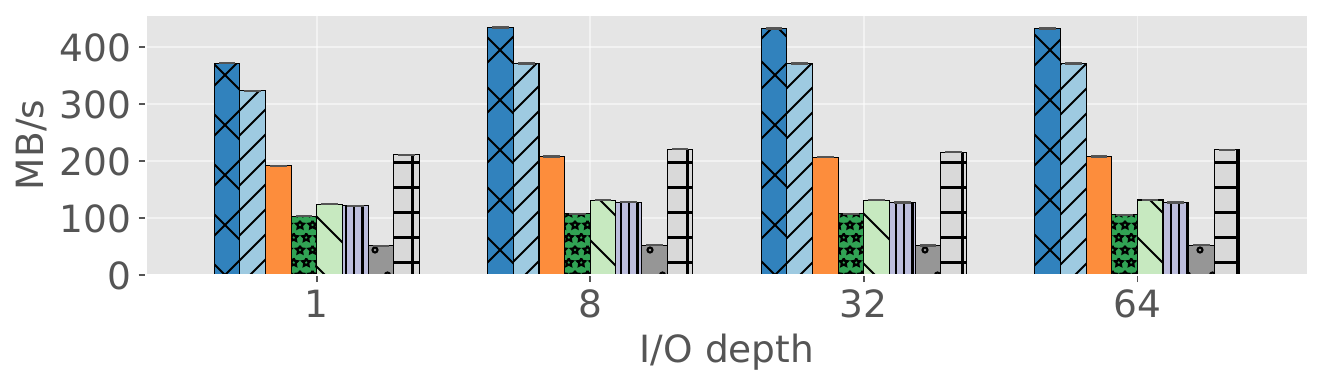}
    \caption{DMTs show speedups across different read ratios, I/O sizes, thread counts, and I/O depths.}
    \label{fig:tp-vs-param}
\end{figure}

\shortsection{Impact of Cache Size} The above analysis showed that DMTs can exploit skewed patterns in workloads when present and deliver a stable performance guarantee across various capacities. Now we examine the effects that other system settings and workload characteristics have on performance. The goal is to evaluate whether these conclusions about DMTs hold broadly. We continue with the Zipf(2.5) workload.

\autoref{fig:perf-vs-cache} shows that DMTs maintain the highest throughputs across both small and large hash caches. Note that cache size is specified as a ratio of the tree size; the absolute cache size varies with capacity. Further, caches mostly benefit read I/Os (because they enable early returns), but write I/Os still must traverse the entire path to the root. In general, we observe that increasing cache size beyond 0.1\% does not yield significant performance improvements for any hash tree design; small caches are already very efficient. Yet, losses observed by all balanced tree designs are still significant. This shows that caching only helps to an extent---when caches are efficient, hash tree overheads are largely attributable to how efficient the tree structure is. However, DMTs still deliver near-optimal performance and the highest across all sizes.

\shortsection{Impact of Read Ratio, I/O Size, Thread Count, and I/O Depth}
\autoref{fig:tp-vs-param} (top) shows how the performance changes with respect to the read ratio. We expect that at higher read ratios, DMTs, balanced, and optimal trees will all observe higher absolute throughputs, as reads can be quickly served by early returns due to caching. However, when there is a significant proportion of writes ($\leq 50\%$ read ratio), DMTs provide nearly $2\times$ higher throughput than balanced trees. 
Since storage access patterns tend to be write-heavy (due to application-level caches and the OS page cache), this shows that DMTs can more reliably handle write-heavy workloads, while delivering comparable performance under read-heavy workloads.

\begin{figure}[t]
    \centering
    \includegraphics[width=0.4\textwidth]{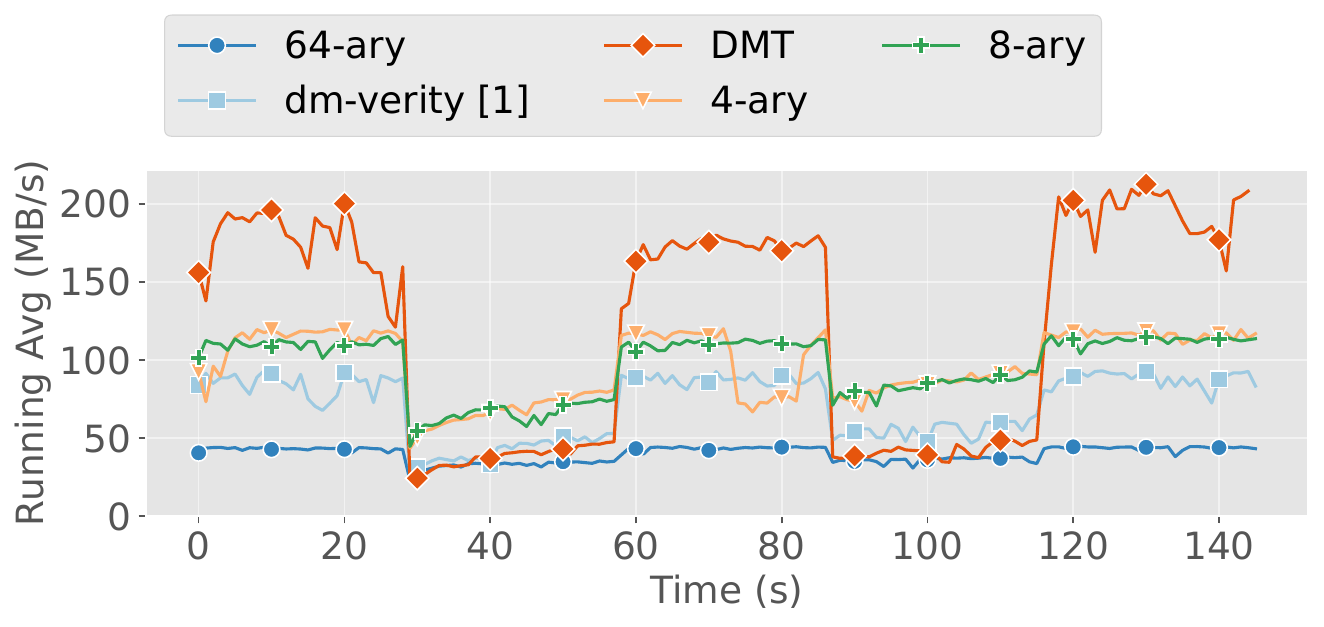}
    \caption{DMTs can adapt quickly to changing workloads, exploiting skewed patterns when present.}
    \label{fig:tp-vs-time}
\end{figure}

The remaining graphs in~\autoref{fig:tp-vs-param} reflect the above observations. Baseline throughputs increase w.r.t. I/O size, but hash tree performance saturates at 32~KB I/Os---larger I/O sizes only lead to increased latencies without improved throughputs. A single thread is sufficient to saturate the device bandwidth. And an application I/O depth of 32 is sufficient to saturate the device bandwidth. DMTs still deliver up to $2\times$ and $4\times$ higher throughput over the state of the art binary and 64-ary trees. The extent to which DMTs see speedups is thus best attributed to the workload shape (\autoref{fig:perf-vs-skew}).

As noted in~\autoref{sec:motivation}, 32~KB write I/Os require 8 sequential hash tree updates. State of the art works are not truly concurrent, but still rely on a global tree lock to serialize tree updates~\cite{taassori2018vault,freij2021bonsai,feng2021scalable}. Performance has been shown to improve under high concurrency by using lazy verification (deferring and batching updates)~\cite{arasu2021fastver}, but lazy verification violates freshness guarantees.
Designing concurrency-optimal hash trees (and search trees in general~\cite{shanny2022occualizer}) is an open problem.

\begin{figure}[t]
    \centering
    \includegraphics[width=0.165\textwidth]{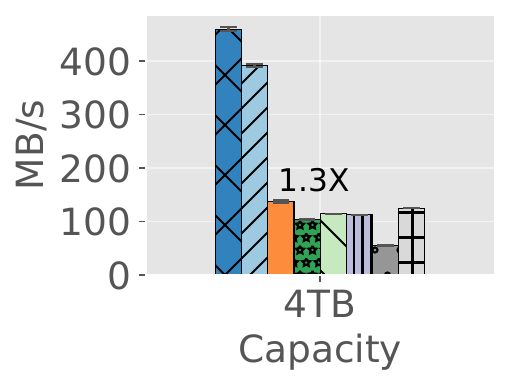}
    \includegraphics[width=0.3\textwidth]{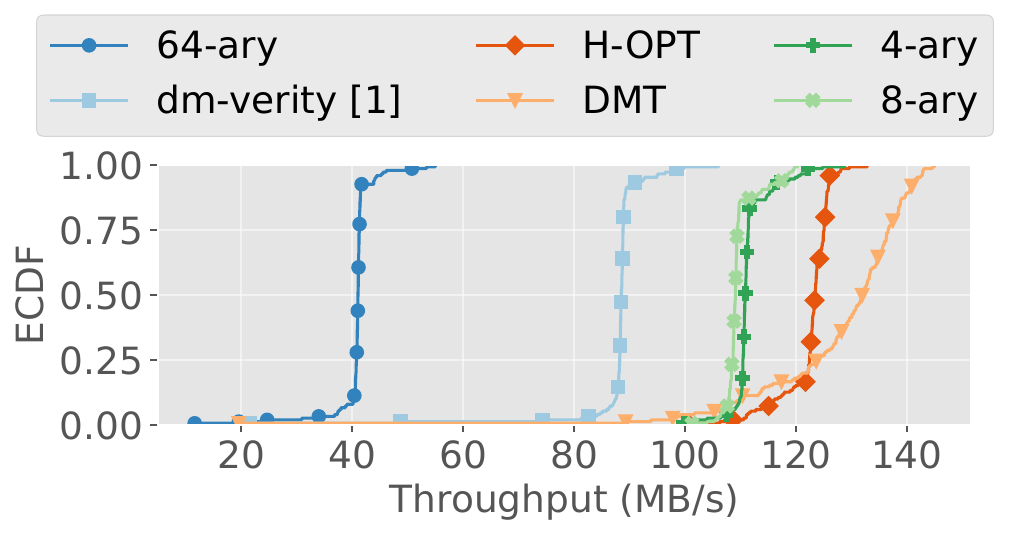}
    \caption{On an Alibaba cloud volume trace, DMTs deliver notable speedups, while high-degree trees perform worst.}
    \label{fig:perf-vs-cap-alibaba}
\end{figure}
\begin{table}[t]
    \centering
    \small
    \begin{tabular}{lrrr}
        \toprule
        & \texttt{DMT} & \texttt{dm-verity} & \texttt{No enc/no integrity} \\
        \midrule
        \textit{write} & 255.4~MB/s & 151.9~MB/s & 318.8~MB/s \\
        \textit{read} & 0.7~MB/s & 0.4~MB/s & 1.0~MB/s \\
        \bottomrule
    \end{tabular}
    \caption{Application read/write throughputs for the Filebench OLTP workload. DMT driver-level improvements are reflected at application-level.}
    \label{tab:oltp}
\end{table}

\shortsection{Handling Changing Access Patterns}
We now demonstrate that DMTs self-adapting nature is robust to changing workload patterns. \autoref{fig:tp-vs-time} shows a 150-second snapshot of sampled throughputs under a workload that exercises an extreme case where patterns alternate between uniform and skewed: Zipf(2.5) > Uniform > Zipf(2.0) > Uniform > Zipf(3.0). Phases are 30 seconds long, and the Zipfian phases are randomly centered at a new region in the address space. We observe that DMT throughput spikes within a few seconds of entering the Zipfian phases and DMTs maintain the speedup throughout. This shows that DMTs can capitalize on skewed patterns very quickly to maximize performance while delivering performance comparable to binary trees otherwise. As noted, extending DMTs to 4-ary trees can help further improve DMT performance during uniform phases. 

\shortsection{Case study: Alibaba Cloud Volumes}
We showed that DMTs perform near-optimally under a broad range of system and workload settings. We now examine how these observations hold under a real workload sampled from a recently published Alibaba trace dataset~\cite{li2023depth} (logical volume ID 4). We scale the offsets and I/O sizes proportionally to the experiment capacity. Note that the remaining volume traces are qualitatively the same (mean write ratio >98\% and highly skewed). Further, note that the workload is non-\textit{i.i.d.} and therefore H-OPT can underestimate the upper bound on throughput; temporal patterns enable DMTs to perform better in some cases.

\autoref{fig:perf-vs-cap-alibaba} (left) shows the aggregate throughputs observed at a 4~TB capacity, and \autoref{fig:perf-vs-cap-alibaba} (right) shows the distribution of write throughputs (sampled at 1-second intervals). Binary trees observe a 75\% throughput loss at 4~TB capacity, while 64-ary trees observe an 88\% throughput loss. DMTs provide a $1.3\times$ speedup over the binary trees and a $1.2\times$ speedup over the 4-ary trees. Importantly, the optimal (binary) tree still observes 15\% higher throughput than 4-ary and 8-ary trees. This further supports our claim that a balanced tree structure is not sufficient to maximize the performance potential. As noted above, we believe the extending DMT principles to a 4-ary tree can help achieve maximum performance.

\shortsection{Case study: OLTP Workload}
We now consider the Filebench OLTP workload, an exemplar application that is commonly run in the cloud and requires robust security protections~\cite{dong2024cloud}. The goal is to evaluate how DMT device-level improvements translate to application-level improvements. The workload consists of 10 writer threads and 200 reader threads and is write-heavy. We run the workload for 10 minutes on a 1~TB disk (with a dataset size of $\approx$922~GB) formatted with ext4 and using a hash cache size of 10\%. \autoref{tab:oltp} shows that DMTs have $1.8\times$ improved read and $1.7\times$ improved write performance over the state of the art. 

\begin{table}[t]
    \centering
    \small
    \begin{tabular}{lcc}
        \toprule
        & \textbf{Memory~Overhead} & \textbf{Storage Overhead} \\
        \midrule
        \textit{leaf nodes} & $0.44\times$ & $0.29\times$ \\
        \textit{internal nodes} & $0.80\times$ & $0.75\times$ \\
        \bottomrule
    \end{tabular}
    \caption{DMTs require additional memory/storage for tree nodes, but break even on this trade-off: they provide higher performance than balanced trees, at a smaller cache budget.}
    \label{tab:memory-overhead}
\end{table}

\shortsection{Memory \& Storage Overhead}
DMTs provide several advantages, but have higher memory and storage requirements than balanced trees (\autoref{tab:memory-overhead}). DMTs cannot use implicit indexing like balanced trees, but instead require explicitly storing parent-child pointers (as integer node IDs) both for nodes in-memory and on-disk. This implies at least one additional integer field for leaf nodes, at most three additional integer fields for internal nodes, and one additional integer hotness counter field for all nodes. However, we showed that cache hit rates are very high even for very small caches. For example, DMTs provide better performance at a cache size of 0.1\% than binary trees do at a cache size of 1\%. Thus, DMTs deliver better performance per dollar spent on cache memory.

\begin{tcolorbox}[colframe=white]
    \textbf{Key takeaways:} State-of-the-art hash trees incur substantial performance loss (up to 80\%). DMTs \newtxt{address} \oldtxt{improve on} this by exploiting skewed \oldtxt{workload} patterns when present and adapting quickly to changes over time. We conclude that balanced \oldtxt{hash} trees are ill-suited as the base construction of a hash tree: maximizing performance requires tailoring the tree structure to \oldtxt{the} workload pattern\newtxt{s}.
\end{tcolorbox}

\section{Related Work}
\label{sec:related}

Hash trees are a core component of many computing systems, including blockchains, secure memories, etc.~\cite{buterin2016ethereum,naor2000certificate,android-dm-verity,taassori2018vault,avanzi2022cryptographic,gassend2003caches,mckeen_innovative_2013,erway2015dynamic,arasu2021fastver,sinha2018veritasdb,li2006dynamic,feng2021scalable}. We discuss these related works below.

\shortsection{Secure Memory}
Secure memories have been long-studied~\cite{yan2006improving,gassend2003caches,feng2021scalable,rogers2007using}. Their goal is to provide a secure environment in which the secrecy and integrity of application code and data can be assured. Recently, secure memories have seen widespread commercial success through implementations such as Intel SGX~\cite{mckeen_innovative_2013,tsai_graphene-sgx_nodate,priebe_sgx-lkl_2020}. 
Consequently, recent work has shown that hash trees can severely degrade memory performance, and optimizing them has been a central focus of recent research. 
State-of-the-art approaches, such as Penglai~\cite{feng2021scalable}, FastVer~\cite{arasu2021fastver}, and VAULT~\cite{taassori2018vault}, have shown that optimizations like caching and increasing tree degree can lower overheads in some contexts (e.g., at small capacities). 

However, we focus on persistent storage rather than main memory. Hash tree performance implications in the context of cloud block storage remain largely unknown, perhaps due to the presumption that hashing costs are negligible at the block layer. We showed that this is not the case. Further, storage systems are subject to vastly different workload characteristics, capacities, and cache behaviors than memory systems. We juxtaposed DMTs against the high-degree trees used for secure memories~\cite{freij2021bonsai,taassori2018vault}) and showed that such approaches are not applicable to storage. Additionally, some prior works have decidedly side-stepped the issue of writes by either suppressing caches during updates or limiting analysis to read-heavy or read-only workloads~\cite{arasu2021fastver,sinha2018veritasdb}. Given that storage workloads are write-heavy, we closely examine write-heavy workloads. 
Further, experiments in these prior works have been limited to small capacities (e.g., 16GB).
We report experiments with an implementation on real block devices on a live Linux system and evaluate a wider range of system settings.

\shortsection{Authenticated Data Structures}
Hash trees have also been examined in the broader theoretical context of authenticated data structures~\cite{tamassia2003authenticated,miller2014authenticated,crosby2011authenticated}. They become a central component of mobile and embedded device storage: dm-verity has played a pivotal role in providing verified boot for Android smartphones~\cite{android-dm-verity}. They have also been examined in the context of blockchains~\cite{buterin2016ethereum}, certificate revocation systems~\cite{melara2015coniks}, and provable data possession schemes~\cite{erway2015dynamic,dahlberg2016efficient}. These works have highlighted the theoretical efficiency of Merkle hash trees.
Our work shows that, while efficient in theory, traditional hash trees still incur substantial overheads in real systems. This motivates our search for a more efficient tree structure and ultimately the design of DMTs.

\section{Conclusion}
\label{sec:conclusion}

\begin{figure}[!t]
    \centering
    \includegraphics[width=0.45\textwidth]{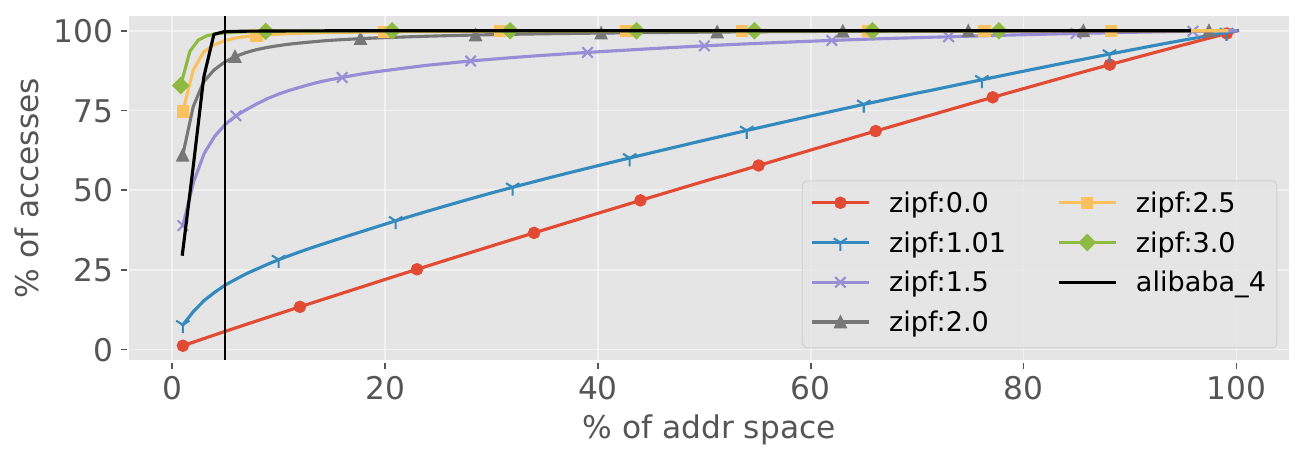}
    \caption{Workload distributions.}
    \label{fig:dist}
\end{figure}

Merkle hash trees provide robust integrity protections over untrusted storage, but they can severely degrade performance. We performed a comprehensive analysis of performance overheads, demonstrated the root cause, and designed an optimized tree structure called DMTs that advance the state of the art. 
DMTs demonstrate the power of integrity structures that exploit workload patterns to reduce overheads. \newtxt{Our code is open-sourced at \href{https://github.com/MadSP-McDaniel/dmt}{https://github.com/MadSP-McDaniel/dmt}}. \oldtxt{Our code is plug-and-play into standard Linux systems and open-sourced at [Anonymized link].}
\section*{Acknowledgments}
This work was supported in part by the Semiconductor Research Corporation (SRC) and DARPA.

\bibliography{main}

\appendix

\ifrevision\input{revision}\fi
\end{document}